\documentclass[a4paper]{amsart}
\usepackage{graphics}

\theoremstyle{plain}
\newtheorem{theorem}{Theorem}
\newtheorem{lemma}[theorem]{Lemma}
\newtheorem{corollary}[theorem]{Corollary}
\newtheorem{proposition}[theorem]{Proposition}
\theoremstyle{definition}
\newtheorem{definition}[theorem]{Definition}
\newtheorem{remark}[theorem]{Remark}

\newtheorem{example}[theorem]{Example}
\newtheorem{problem}{Problem}

\DeclareMathOperator{\val}{val}
\DeclareMathOperator{\rep}{rep}
\DeclareMathOperator{\lcm}{lcm}

\title[A decision problem \ldots in non-standard numeration systems]
{A decision problem for ultimately periodic sets in non-standard numeration systems}

\author[J. Bell]{Jason Bell}
\author[E. Charlier]{Emilie Charlier}
\author[A. S. Fraenkel]{Aviezri S. Fraenkel}
\author[M. Rigo]{Michel Rigo}

\address[J. Bell]{Department of Mathematics, Simon Fraser University
 Burnaby, BC, CANADA V5A 1S6}
\email{jpb@sfu.ca}

\address[E. Charlier, M. Rigo]{Institute of Mathematics, University of Li\`ege, 
Grande Traverse 12 (B 37), B-4000 Li\`ege, Belgium}
\email{\{echarlier,M.Rigo\}@ulg.ac.be}

\address[A. S. Fraenkel]{
Department of Computer Science \& Applied Mathematics,  
Weizmann Institute of Science, 
76100 Rehovot, 
Israel.}
\email{aviezri.fraenkel@weizmann.ac.il}

\begin{document}
\maketitle

\begin{abstract}
    Consider a non-standard numeration system like the one built over
    the Fibonacci sequence where nonnegative integers are represented
    by words over $\{0,1\}$ without two consecutive $1$. Given a set
    $X$ of integers such that the language of their greedy
    representations in this system is accepted by a finite automaton,
    we consider the problem of deciding whether or not $X$ is a finite
    union of arithmetic progressions. We obtain a decision procedure
    for this problem, under some hypothesis about the considered
    numeration system. In a second part, we obtain an analogous
    decision result for a particular class of abstract numeration
    systems built on an infinite regular language.
\end{abstract}

\section{Introduction}
\begin{definition}\label{def:numsys}
    A {\em positional numeration system} is given by a (strictly) increasing
    sequence $U=(U_i)_{i\ge 0}$ of integers such that $U_0=1$ and
    $C_U:=\sup_{i\ge 0} \lceil U_{i+1}/U_i\rceil$ is finite. Let
    $A_U=\{0,\ldots,C_U-1\}$. The {\em greedy $U$-representation} of a
    positive integer $n$ is the unique finite word
    $\rep_U(n)=w_\ell\cdots w_0$ over $A_U$ satisfying
$$n=\sum_{i=0}^\ell w_i\, U_i,\ w_\ell\neq 0 \text{ and } \sum_{i=0}^t w_i\,
U_i<U_{t+1},\ \forall t=0,\ldots,\ell.$$ We set $\rep_U(0)$ to be the
empty word $\varepsilon$. A set $X\subseteq\mathbb{N}$ of integers is
{\it $U$-recognizable} if the language $\rep_U(X)$ over $A_U$ is
regular (i.e., accepted by a deterministic finite automaton, DFA). If $x=x_\ell\cdots
x_0$ is a word over a finite alphabet of 
integers, then
the {\it $U$-numerical value} of $x$ is $$\val_U(x)=\sum_{i=0}^\ell
x_i\, U_i.$$
\end{definition}

\begin{remark}
    Let $x,y$ be two words over $A_U$. As a consequence of the greediness of the representation, if $xy$
    is a greedy $U$-representation and if the leftmost letter of $y$ is
    not $0$, then $y$ is also a greedy $U$-representation. Notice that
    for $m,n\in\mathbb{N}$, we have $m<n$ if and only if
    $\rep_U(m)<_{gen}\rep_U(n)$ where $<_{gen}$ is the genealogical
    ordering over $A_U^*$: words are ordered by increasing length and
    for words of same length, one uses the lexicographical ordering
    induced by the natural ordering of the digits in the alphabet
    $A_U$. Recall that for two words $x,y\in A_U^*$ of same length,
    $x$ is lexicographically smaller than $y$ if there exist
    $w,x',y'\in A_U^*$ and $a,b\in A_U$ such that $x=wax'$,
    $y=wby'$ and $a<b$.
\end{remark}

For a positional numeration system $U$, it is natural to expect that
$\mathbb{N}$ is $U$-recognizable. A necessary condition is that the
sequence $U$ satisfies a linear recurrence relation \cite{Sha}.

\begin{definition}
    A positional numeration system $U=(U_i)_{i\ge 0}$ is said to be
    {\em linear}, if the sequence $U$ satisfies a homogenous linear
    recurrence relation with integer coefficients. Otherwise stated,
    there exist $k\ge 1$ and constant coefficients $a_1,\ldots,a_k$
    such that for all $i\ge 0$, we have
    \begin{equation}
        \label{eq:linrec}
        U_{i+k}=a_1 U_{i+k-1}+\cdots +a_k U_i,\quad\text{with } a_1,\ldots,a_k\in\mathbb{Z},\ a_k\neq 0.
    \end{equation}
    We say that $k$ is the {\em order} of the recurrence relation.
\end{definition}

\begin{example}
    Consider the sequence defined by $F_0=1$, $F_1=2$ and for all
    $i\ge 0$, $F_{i+2}=F_{i+1}+F_{i}$. The {\it Fibonacci (linear
      numeration) system} is given by $F=(F_i)_{i\ge
      0}=(1,2,3,5,8,13,\ldots)$.  For instance, $\rep_F(15)=100010$
    and $\val_F(101001)=13+5+1=19$.
\end{example}

In this paper, we mainly address the following decidability question
and its extension to abstract numeration systems.
\begin{problem}\label{pb}
    Given a linear numeration system $U$ and a set
    $X\subseteq\mathbb{N}$ such that $\rep_U(\mathbb{N})$ and
    $\rep_U(X)$ are both recognized by (deterministic) finite
    automata. Is it decidable whether or not $X$ is ultimately
    periodic, i.e., whether or not $X$ is a finite union of arithmetic
    progressions ?
\end{problem}

Notice that the regularity of $\rep_U(\mathbb{N})$ ensures that there
exists a set $X\subseteq\mathbb{N}$ such that $\rep_U(X)$ is
regular, see also Remark~\ref{rem:nurec}.\smallskip

Ultimately periodic sets of integers play a special role.  On the one
hand such infinite sets are coded by a finite amount of
information. On the other hand the celebrated Cobham's theorem asserts
that these sets are the only sets that are recognizable in all integer
base systems \cite{Cob}. This is the reason why they are also referred
in the literature as {\it recognizable} sets of integers (the
recognizability being in that case independent of the base). Moreover,
Cobham's theorem has been extended to various situations and in
particular, to numeration systems given by substitutions \cite{Dur}.
\smallskip

If we restrict ourselves to the usual integer base $b\ge 2$ numeration
system, referred to in the sequel as {\em $b$-ary system}, defined by
$U_i=b\, U_{i-1}$ for $i\ge 1$ and $U_0=1$, several results are known.
J.~Honkala showed in \cite{Hon} that Problem~\ref{pb} turns out to be
decidable. Let us also mention \cite{min}, where the number of states
of the minimal automaton accepting numbers written in base $b$ and
divisible by $d$ is given explicitly. J.-P.~Allouche and J.~Shallit
ask in \cite{AS} if one can obtain a polynomial time decision
procedure for $b$-ary systems. Using the logic formalism of the
Presburger arithmetic, a positive answer to this question is given by
J.~Leroux in \cite{Leroux} even when considering subsets of
$\mathbb{Z}^d$, $d\ge 1$.  In dimension one, ultimately periodic sets
are exactly the sets definable in the Presburger arithmetic
$\langle\mathbb{N},+\rangle$.  \smallskip

Note that A.~Muchnik has shown that Problem~\ref{pb} turns out to be
decidable for any linear numeration system $U$ for which both
$\rep_U(\mathbb{N})$ and addition are recognizable by automata
\cite{much}. But it is a difficult question to characterize numeration
systems $U$ for which addition is {\em computable by finite
  automaton}, i.e., for which the language
$$
\biggl\{\begin{pmatrix}
    0^{m-|\rep_U(x)|}\rep_U(x)\\
    0^{m-|\rep_U(y)|}\rep_U(y)\\
    0^{m-|\rep_U(z)|}\rep_U(z)\\
\end{pmatrix}
\, \mid \, x,y,z\in\mathbb{N},\ x+y=z,\ 
m=\max_{t\in\{x,y,z\}}|\rep_U(t)|\biggr\} $$
where leading zeroes are
prefixed for padding the shorter components to obtain three words of
the same length, is regular (see for instance \cite{BH,Fro} in which
the characteristic polynomial of the sequence $U$ is the minimal
polynomial of a Pisot number). In \cite{FroSeq}, the sequentiality of
the successor function (i.e., the action of adding $1$) is studied. If
addition is computable by a finite automaton, so is the successor
function, but the converse does not hold in general. In particular,
some examples of linear numeration systems for which addition is not
computable by a finite automaton are given in \cite{FroSeq}: for
instance, the sequence defined by $U_i=3U_{i-1}+2U_{i-2}+3U_{i-4}$
with integer initial conditions $1=U_0<U_1<U_2<U_3$. So the decision
techniques from \cite{Leroux,much} cannot be applied to that system.
Nevertheless, as we will see in Example~\ref{exa:aaa}, our decision
procedure can be applied to this system. Notice also that in the
extended framework of abstract numeration systems, one can exhibit
systems such that multiplication by a constant does not preserve
recognizability \cite{CRS,LR,R}.  Therefore the powerful tools from
logic discussed above cannot be applied in that context either.
\smallskip

The question studied in this paper was raised by
J.~Sakarovitch during the ``{\it Journ\'ees de Num\'eration}'' in
Graz, May 2007. The question was initially asked for a larger class of
systems than the one treated here, namely for any abstract numeration
systems defined on an infinite regular language \cite{LR}. A shorter
and partial version of this paper has been presented to the 33rd MFCS
symposium in Toru\'n, August 2008 \cite{CR}. \smallskip

Observe that this decision problem for all abstract numeration systems
is equivalent to the famous {\it HD0L periodicity problem}: given a
morphism $f$ and a coding $g$, decide if the infinite word
$g(f^\omega(a))$ is ultimately periodic, see \cite{HR,RM}. For the
restricted case of the {\it D0L periodicity problem}, where only the
morphism $f$ is considered, decision procedures are well-known
\cite{HL,Pansiot}.  Finally, questions connected to the ones addressed
here have independently and recently gained interest \cite{ARS}. In
particular, a simple proof of Honkala's original result based on the
construction of some automata is given in that paper. As for the
logical approach considered by Muchnik and Leroux, the arguments given
in \cite{ARS} rely on the recognizability of addition by automata
(which can be done for the classical $k$-ary system but not
necessarily for an arbitrary linear numeration system).  \smallskip

The structure of this paper is the same as \cite{Hon}. First we give
an upper bound on the admissible periods of a $U$-recognizable set $X$
when it is assumed to be ultimately periodic. Then an upper bound on
the admissible preperiods is obtained. These bounds depend essentially
on the number of states of the (minimal) automaton recognizing
$\rep_U(X)$. Finally, finitely many such periods and preperiods have
to be checked. For each of them, we have to build an automaton
accepting the corresponding ultimately periodic sets (this implies
that $\mathbb{N}$ has to be recognizable).  

Though the structure is the same, our arguments and techniques are
quite different from \cite{Hon}. They rely on the study of the
quantity $N_U(m)$ defined as the number of residue classes that appear
infinitely often in the sequence $(U_i\mod m)_{i\ge 0}$. Our main
result --- Theorem~\ref{the:finfin} --- can be stated as follows. 
\smallskip

\noindent
{\bf Theorem.} Let $U=(U_i)_{i\ge 0}$ be a linear numeration system
such that $\mathbb{N}$ is $U$-recognizable and satisfying
$\lim_{i\to+\infty} U_{i+1}-U_i=+\infty$. If
$\lim_{m\to+\infty}N_U(m)=+\infty$, then Problem~\ref{pb} is
decidable.  \smallskip

Actually our techniques cannot be applied to $b$-ary 
systems (see Remark~\ref{rem:notint}), which is the case treated by Honkala
\cite{Hon}, because in that case
$N_U(m)\not\to+\infty$ as $m\to+\infty$.

In Section~\ref{sec:lrrc}, we give a characterization of the linear
numeration systems $U$ such that $\lim_{m\to+\infty}N_U(m)=+\infty$.
To do so we use $p$-adic methods leading to a study of the
sequence $(U_i\mod p^v)_{i\ge 0}$ for all $v\ge 1$, where $p$ is a
prime dividing $a_k$.

In the last section, we consider again the same decision problem but
restated in the framework of abstract numeration systems \cite{LR}.
The definition of these systems is given in Section~\ref{sec:last}. We
apply successfully the same kind of techniques to a large class of
abstract numeration systems (for instance, an example consisting of
two copies of the Fibonacci system is considered).  The corresponding
decision procedure is given by Theorem~\ref{the:ans}. As explained
above, this result provides a decision procedure for specific
instances of the HD0L periodicity problem.

All along the paper, we try whenever it is possible to state results
in their most general form, even if later on we have to restrict
ourselves to particular cases. For instance, results about the
admissible preperiods do not require any particular assumption on the
numeration system except linearity.

\section{A Decision Procedure for a Class of Linear Numeration Systems}

We will often consider positional numeration systems $U=(U_i)_{i\ge 0}$
satisfying the following condition: \begin{equation}
        \label{eq:condlim}
        \lim_{i\to+\infty} U_{i+1}-U_i=+\infty.
    \end{equation}
    Notice that it is a weak requirement. Usually, the sequence $U$
    has an exponential growth, $U_i\simeq \beta^i$ for some $\beta>1$,
    and therefore \eqref{eq:condlim} is trivially satisfied. (It is
    for instance the case for the numeration systems considered in
    Remark~\ref{rem:exxx}.)

\begin{lemma}\label{lem:condlim}
    Let $U=(U_i)_{i\ge 0}$ be a positional numeration system
    satisfying \eqref{eq:condlim}.  Then for all $j$, there exists
    $L$ such that for all $\ell\ge L$,
$$1 0^{\ell-|\rep_U(t)|}\rep_U(t),\ t=0,\ldots,U_j-1$$
are greedy $U$-representations. Otherwise stated, if $w$ is a greedy $U$-representation, 
then for $r$ large enough, $10^rw$ is also a greedy $U$-representation.
\end{lemma}

\begin{proof}
    Notice that $\rep_U(U_j-1)$ is the greatest word of length $j$ in
    $\rep_U(\mathbb{N})$, since $\rep_U(U_j)=10^j$. 
    By hypothesis, there exists $L$ such that for all $\ell\ge L,$ 
    $U_{\ell+1}-U_\ell> U_j-1$. 
    Therefore, for all $\ell\ge L,$
    $$10^{\ell-j} \rep_U(U_j-1)$$ is the greedy
    $U$-representation of $U_\ell+U_j-1<U_{\ell+1}$ and the conclusion follows.
\end{proof}

\begin{example}
    Consider the positional numeration system $U_0=1$, $U_1=2$, $U_2=3$ and
    $U_{3i+r}=3^{i+1}+r$ for all $i\ge 1$ and $r\in\{0,1,2\}$. This
    system does not satisfy \eqref{eq:condlim} because $U_{i+1}-U_i=1$
    for infinitely many $i$'s. We have $\rep_U(2)=10$, but one can
    notice that for $i\ge 0$, $10^{3i+1}10$ is not a greedy
    $U$-representation. Indeed,
    $\val_U(10^{3i+1}10)=U_{3(i+1)}+2=U_{3i+5}$ has $10^{3i+5}$
    as greedy $U$-representation.
\end{example}

\begin{remark}
    In the above lemma, one cannot exchange the order of the
    quantifiers about $j$ and $L$. As an example, consider the
    sequence $U_i=(i+1)(i+2)/2$ for all $i\ge 0$.  This sequence
    satisfies the linear recurrence relation
    $U_{i+3}=3U_{i+2}-3U_{i+1}+U_i$ and also \eqref{eq:condlim}.
    Observe that, for all $i\ge 1$, $k=U_i-1$ is the unique value such
    that $U_i=U_k-U_{k-1}$ because for all $j\ge 1$,
    $U_j-U_{j-1}=j+1$. For all $i\ge 1$, $10^i$ is a greedy
    $U$-representation and the greedy $U$-representations of
    the form $10^n10^i$ are exactly those for which $n\ge U_i-i-1$.
\end{remark}

\begin{remark} Bertrand numeration systems associated with a real
    number $\beta>1$ are defined as follows. Let
    $A_\beta=\{0,\ldots,\lceil\beta\rceil-1\}$. Any $x\in[0,1]$ can be written as 
$$x=\sum_{i=1}^{+\infty} c_i\, \beta^{-i}, \text{ with }c_i\in A_\beta$$ 
and the sequence $(c_i)_{i\ge 1}$ is said to be a {\it
  $\beta$-representation} of $x$. The maximal $\beta$-representation
of $x$ for the lexicographical order is denoted $d_\beta(x)$ and is
called the {\it $\beta$-development} of $x$ (for details see
\cite[Chap. 8]{Lot}). We say that a $\beta$-development $(c_i)_{i\ge
  1}$ is {\it finite} if there exists $N$ such that $c_i=0$ for all
$i\ge N$.  If there exists $m\ge 1$ such that $d_\beta(1)=t_1\cdots t_m$ 
with $t_m\neq 0$, we set
$d_\beta^*(1):=(t_1\cdots t_{m-1}(t_m-1))^\omega$, otherwise
$d_\beta(1)$ is infinite and we set $d_\beta^*(1):=d_\beta(1)$.

We can now define a positional numeration system $U_\beta=(U_i)_{i\ge 0}$
associated with $\beta$ (see \cite{AB}). If $d_\beta^*(1)=(t_i)_{i\ge
  1}$, then
\begin{equation}
    \label{eq:betaa}
U_0=1\text{ and }\forall i\ge 1,\ U_i=t_1U_{i-1}+\cdots +t_iU_0+1.    
\end{equation}
If $\beta$ is a Parry number (i.e., $d_\beta(1)$ is finite or
ultimately periodic) then one can derive from \eqref{eq:betaa} that
the sequence $U_\beta$ satisfies a linear recurrence relation and as a
consequence of Bertrand's theorem \cite{AB} linking greedy
$U_\beta$-representations and finite factors occurring in
$\beta$-developments, the language $\rep_{U_\beta}(\mathbb{N})$ of the
greedy $U_\beta$-representations is regular. The automaton accepting
these representations is well-known \cite{FS} and has a special form
(all states --- except for a sink --- are final and from all these
states, an edge of label $0$ goes back to the initial state). We
therefore have the following property which is much stronger than the
previous lemma. If $x$ and $y$ are greedy $U_\beta$-representations
then $x0y$ is also a greedy $U_\beta$-representation.
\end{remark}

\begin{example}
    The Fibonacci system is the Bertrand system associated with the
    golden ratio $(1+\sqrt{5})/2$. Since greedy representations in
    the Fibonacci system are the words not containing two consecutive $1$'s
    \cite{zec}, then for $x,y\in\rep_F(\mathbb{N})$, we have
    $x0y\in\rep_F(\mathbb{N})$.
\end{example}

\begin{definition}
    Let $X\subseteq\mathbb{N}$ be a set of integers. The {\it
      characteristic word of $X$} is an infinite word
    $x_0x_1x_2\cdots$ over $\{0,1\}$ defined by $x_i=1$ if and only if
    $i\in X$.

    Consider for now $X\subseteq\mathbb{N}$ to be an ultimately
    periodic set. The characteristic word of $X$ is therefore an
    infinite word over $\{0,1\}$ of the form $$x_0x_1x_2\cdots =
    uv^\omega$$ 
    where $u$ and $v$ are chosen of minimal length. We say
    that the length $|u|$ of $u$ (resp. the length $|v|$ of $v$) is
    the {\em preperiod} (resp. {\em period}) of $X$. Hence, for all
    $n\ge|u|$, $n\in X$ if and only if $n+|v|\in X$.
\end{definition}

The following lemma is a simple consequence of the minimality of the
period chosen to represent an ultimately periodic set.

\begin{lemma}\label{lem:per}
    Let $X\subseteq\mathbb{N}$ be an ultimately periodic set of period
    $p_X$ and preperiod $a_X$. Let $i,j\ge a_X$. If $i\not\equiv j\mod
    p_X$ then there exists $t<p_X$ such that either $i+t\in X$ and
    $j+t\not\in X$ or, $i+t\not\in X$ and $j+t\in X$.
\end{lemma}

We assume that the reader is familiar with automata theory (see for
instance \cite{Saka}) but let us recall some classical results. Let
$L\subseteq\Sigma^*$ be a language over a finite alphabet $\Sigma$ and
$x$ be a finite word over $\Sigma$. We set
$$x^{-1}L=\{z\in\Sigma^*\mid xz\in L\}.$$
We can now define the Myhill-Nerode congruence. Let $x,y\in\Sigma^*$.
We have $x\sim_L y$ if and only if $x^{-1}L=y^{-1}L$. Moreover $L$
is regular if and only if $\sim_L$ has a finite index being the number
of states of the minimal automaton of $L$.
\medskip

\begin{definition}
For a sequence $(U_i)_{i\ge0}$ of integers,
$N_U(m)\in\{1,\ldots,m\}$ denotes the number of values that are
taken infinitely often by the sequence $(U_i\mod m)_{i\ge0}$.
\end{definition}

\begin{proposition}\label{pro:1}
    Let $U=(U_i)_{i\ge 0}$ be a positional numeration system
    satisfying \eqref{eq:condlim}.  If $X\subseteq\mathbb{N}$ is an
    ultimately periodic $U$-recognizable set of period $p_X$, then any
    deterministic finite automaton accepting $\rep_U(X)$ has at least
    $N_U(p_X)$ states.
\end{proposition}

\begin{proof}
    Let $a_X$ be the preperiod of $X$. By Lemma~\ref{lem:condlim},
    there exists $L$ such that for any $h\ge L$, the words
$$10^{h-|\rep_U(t)|}\rep_U(t),\ t=0,\ldots,p_X-1$$
are greedy $U$-representations. The sequence $(U_i\mod p_X)_{i\ge 0}$
takes infinitely often $N:=N_U(p_X)$ different values. Let
$h_1,\ldots,h_{N}\ge L$ be such that
$$i\neq j\Rightarrow U_{h_i}\not\equiv U_{h_j}\mod p_X$$
and $h_1,\ldots,h_{N}$ can be chosen such that $U_{h_i}>a_X$ for
all $i\in\{1,\ldots,N\}$.
\medskip

By Lemma~\ref{lem:per}, for all $i,j\in\{1,\ldots,N\}$ such that
$i\neq j$, there exists $t_{i,j}<p_X$ such that either
$U_{h_i}+t_{i,j}\in X$ and $U_{h_j}+t_{i,j}\not\in X$, or
$U_{h_i}+t_{i,j}\not\in X$ and $U_{h_j}+t_{i,j}\in X$. Therefore,
$$w_{i,j}=0^{|\rep_U(p_X-1)|-|\rep_U(t_{i,j})|}\rep_U(t_{i,j})$$ is a
word such that either
$$10^{h_i-|\rep_U(p_X-1)|}w_{i,j}\in\rep_U(X) \text{ and } 
10^{h_j-|\rep_U(p_X-1)|}w_{i,j}\not\in\rep_U(X),$$
or
$$10^{h_i-|\rep_U(p_X-1)|}w_{i,j}\not\in\rep_U(X) \text{ and } 
10^{h_j-|\rep_U(p_X-1)|}w_{i,j}\in\rep_U(X).$$ Therefore the words
$10^{h_1-|\rep_U(p_X-1)|},\ldots,10^{h_N-|\rep_U(p_X-1)|}$ are
pairwise non\-equi\-valent for the relation $\sim_{\rep_U(X)}$ and the
minimal automaton of $\rep_U(X)$ has at least $N=N_U(p_X)$ states.
\end{proof}

The previous proposition has an immediate consequence for getting a bound
on the period of a periodic set accepted by a given DFA.

\begin{corollary}\label{cor:2}
    Let $U=(U_i)_{i\ge 0}$ be a positional numeration system satisfying 
    \eqref{eq:condlim}.  Assume that 
    $$\lim_{m\to+\infty}N_U(m)=+\infty.$$ Then the period of an
    ultimately periodic set $X\subseteq\mathbb{N}$ such that
    $\rep_U(X)$ is accepted by a DFA with $d$ states is bounded by the
    smallest integer $s_0$ such that for all $m\ge s_0$,
    $N_U(m)>d$.
\end{corollary}

A result similar to the previous corollary (in the sense that it
permits to give an upper bound on the period) can be stated as
follows. This result will be used later on in our decision procedure
to compute some explicit estimate.

\begin{proposition}\label{pro:alternative}
    Let $U=(U_i)_{i\ge 0}$ be a positional numeration system
    satisfying \eqref{eq:condlim} and $X\subseteq\mathbb{N}$ be an
    ultimately periodic $U$-recognizable set of period $p_X$. Let $c$
    be a divisor of $p_X$. If $1$ occurs infinitely many times in
    $(U_i\mod c)_{i\ge 0}$ then any deterministic finite automaton
    accepting $\rep_U(X)$ has at least $c$ states.
\end{proposition}

\begin{proof} Let $a_X$ be the  preperiod of $X$. 
    Applying several times Lemma~\ref{lem:condlim}, there exist
    $n_1,\ldots,n_c$ such that
    $$10^{n_c}10^{n_{c-1}}\cdots
    10^{n_1}0^{|\rep_U(p_X-1)|-|\rep_U(t)|}\rep_U(t),\,t=0,\ldots,p_X-1$$
    are greedy $U$-representations. Moreover, since $1$ occurs
    infinitely many times in the sequence $(U_i\mod c)_{i\ge 0}$,
    $n_1,\ldots,n_c$ can be chosen such that, for all
    $j=1,\ldots,c$,
                $$\val_U(10^{n_j}\cdots 10^{n_1+|\rep_U(p_X-1)|})\equiv j \mod c$$
                and
                $$\val_U(10^{n_1+|\rep_U(p_X-1)|})>a_X.$$
                For $i,j\in\{1,\ldots,c\}$, $i\neq j$, by
                Lemma~\ref{lem:per} and since $c$ divides $p_X$, the words
                $$10^{n_i}\cdots 10^{n_1} \text{ and } 10^{n_j}\cdots 10^{n_1}$$
                are nonequivalent for $\sim_{\rep_U(X)}$. This can be
                shown by concatenating some word of the kind
                $0^{|\rep_U(p_X-1)|-|\rep_U(t)|}\rep_U(t)$ with
                $t<p_X$, as in the proof of Proposition~\ref{pro:1}.
                This concludes the proof.
\end{proof}

\begin{definition}
For a sequence $(U_i)_{i\ge 0}$ of integers, if $(U_i\mod m)_{i\ge 0}$
is ultimately periodic, we denote its (minimal) preperiod by
$\iota_U(m)$ (we choose the notation $\iota$ to allude to the word index
which is equally used as preperiod) and its (minimal) period by
$\pi_U(m)$. 
\end{definition}

\begin{remark}\label{rem:nupiu}
    Observe that for any linear recurrence sequence of order $k$
    satisfying \eqref{eq:linrec}, we have
    $$N_U(m)\le \pi_U(m)\le (N_U(m))^k.$$
    Therefore,
    $\lim_{m\to+\infty} N_U(m)=+\infty$ if and only if
    $\lim_{m\to+\infty} \pi_U(m)=+\infty$.  Notice that if $m=p.q$
    with $\gcd(p,q)=1$, then $\pi_U(m)=\lcm\{\pi_U(p),\pi_U(q)\}$.
\end{remark}

Now we want to obtain an upper bound on the preperiod of any
ultimately periodic $U$-recognizable set recognized by a given DFA.

\begin{proposition}\label{pro:preper}
    Let $U=(U_i)_{i\ge 0}$ be a linear numeration system. Let
    $X\subseteq\mathbb{N}$ be an ultimately periodic $U$-recognizable
    set of period $p_X$ and preperiod $a_X$. Then any deterministic finite
    automaton accepting $\rep_U(X)$ has at least
    $|\rep_U(a_X-1)|-\iota_U(p_X)$ states.
\end{proposition}

The arguments of the following proof are similar to the one found in
\cite{Hon}.

\begin{proof}  W.l.o.g. we can assume that
    $|\rep_U(a_X-1)|-\iota_U(p_X)>0$.  The sequence $(U_i\mod
    p_X)_{i\ge 0}$ is ultimately periodic with preperiod
    $\iota_U(p_X)$ and period $\pi_U(p_X)$.  Proceed by contradiction
    and assume that $\mathcal{A}$ is a deterministic finite automaton
    with less than $|\rep_U(a_X-1)|-\iota_U(p_X)$ states accepting
    $\rep_U(X)$.  There exist words $w,w_4$ such that the greedy
    $U$-representation of $a_X-1$ can be factorized as
                $$\rep_U(a_X-1)=ww_4$$
                with $|w|=|\rep_U(a_X-1)|-\iota_U(p_X)$. By the
                pumping lemma, $w$ can be written $w_1w_2w_3$ with
                $w_2\neq\varepsilon$ and for all $i\ge 0$,
                $$w_1w_2^iw_3w_4\in\rep_U(X)\Leftrightarrow
                w_1w_2w_3w_4\in\rep_U(X).$$
                By minimality of $a_X$ and
                $p_X$, either $a_X-1\in X$ and for all $n\ge 1$,
                $a_X+np_X-1\not\in X$, or $a_X-1\not\in X$ and for all
                $n\ge 1$, $a_X+np_X-1\in X$.  Using the ultimate
                periodicity of $(U_i\mod p_X)_{i\ge 0}$, we observe
                that, because $|w_4|=\iota_U(p_X)$, for all $i\ge 0$, we have
                $$\val_U(w_1w_2^{i\pi_U(p_X)}w_2w_3w_4)\equiv\val_U(w_1w_2w_3w_4)+
                i\val_U(w_2^{\pi_U(p_X)}0^{|w_2w_3w_4|})\mod p_X.$$
Therefore repeating a factor
                of length multiple of $\pi_U(p_X)$ exactly $p_X$ times
                does not change the value mod $p_X$ and we get
                $$\val_U(w_1w_2^{p_X\pi_U(p_X)}w_2w_3w_4)\equiv\val_U(w_1w_2w_3w_4)\mod p_X,$$
                leading to a contradiction. 
\end{proof}

For the sake of completeness, we restate some well-known properties of
ultimately periodic sets (see for instance \cite{Saka} for a prologue
on Pascal's machine for $b$-ary systems).

\begin{lemma}\label{lem:upU}
    Let $a,b$ be nonnegative integers and $U=(U_i)_{i\ge 0}$ be a
    linear numeration system. The language
$$\val_U^{-1}(a\mathbb{N}+b)=\{w\in A_U^*\mid \val_U(w)\in a\mathbb{N}+b\}\subset A_U^*$$
is regular. In particular, if $\mathbb{N}$ is $U$-recognizable then a
DFA accepting $\rep_U(a\mathbb{N}+b)$ can be obtained efficiently and
any ultimately periodic set is $U$-recognizable.
\end{lemma}

Before giving the proof, notice that for any integer $n\ge 0$,
$\val_U^{-1}(n)\setminus 0^+A_U^*$ is a finite set of words $\{x_1,\ldots,x_{t_n}\}$ over
$A_U$ such that $\val_U(x_i)=n$ for all $i=1,\ldots,t_n$. This set
contains in particular $\rep_U(n)$.

\begin{proof}
    Since regular sets are stable under finite modification (i.e.,
    adding or removing a finite number of words in the language), we can
    assume that $0\le b<a$. The sequence $(U_i\mod a)_{i\ge 0}$ is
    ultimately periodic with preperiod $\ell=\iota_U(a)$ and period
    $p=\pi_U(a)$. It is an easy exercise to build a deterministic
    finite automaton $\mathcal{A}$ accepting the reversal of the words in
    $\{w\in A_U^*\mid \val_U(w)\in a\mathbb{N}+b\}$. The alphabet of
    the automaton is $A_U$. States are pairs $(r,s)$ where $0\le r<a$
    and $0\le s<\ell+p$. The initial state is $(0,0)$.  Final states
    are the ones with the first component equal to $b$.  Transitions
    are defined as follows
$$\forall s<\ell+p-1:\ (r,s)\stackrel{j}{\longrightarrow} (jU_s+r\mod a,\ s+1)$$
$$(r,\ell+p-1)\stackrel{j}{\longrightarrow} (jU_{\ell+p-1}+r\mod a,\ \ell),$$
for all $j\in A_U$. Notice that $\mathcal{A}$ does not check the
greediness of the accepted words, the construction only relies on the
$U$-numerical value of the words modulo $a$.

For the particular case, one has to consider the intersection of two
regular languages $\rep_U(\mathbb{N})\cap\val_U^{-1}(a\mathbb{N}+b)$.
\end{proof}

\begin{remark}\label{rem:nurec}
    In the previous statement, the assumption about the
    $U$-recognizability of $\mathbb{N}$ is of particular interest.
    Indeed, it is well-known that for an arbitrary linear numeration
    system, $\mathbb{N}$ is in general {\it not} $U$-recognizable. If
    $\mathbb{N}$ is $U$-recognizable, then $U$ satisfies a linear
    recurrence relation \cite{Sha}, but the converse does not hold.
    Sufficient conditions on the recurrence relation that $U$
    satisfies for $\mathbb{N}$ to be $U$-recognizable are given in
    \cite{Ho}.
\end{remark}

Our decision procedure will also make use of the following result.

\begin{lemma}\label{lem:equiv}
    Let $U=(U_i)_{i\ge 0}$ be an increasing sequence satisfying a linear 
    recurrence relation of order $k$ of the kind \eqref{eq:linrec}.
    The following assertions are equivalent:
    \begin{itemize}
      \item[(i)] $\lim_{m\to+\infty}N_U(m)=+\infty$
      \item[(ii)] for all prime divisors $p$ of $a_k$, $\lim_{v\to+\infty}N_U(p^v)=+\infty$.
    \end{itemize}
    In particular, if $a_k=\pm 1$, then $\lim_{m\to+\infty}N_U(m)=+\infty$.
\end{lemma}

\begin{proof} It is enough to show that (ii) implies (i). 
    Let the prime decomposition of $|a_k|$ be $|a_k|=p_1^{u_1}\cdots
    p_r^{u_r}$ with $u_1,\ldots,u_r> 0$. It is obvious that if
    $m=p_1^{v_1}\cdots p_r^{v_r}c$ with $v_1,\ldots,v_r\ge0$ and
    $\gcd(a_k,c)=1$ then
    \begin{equation*}
    \pi_U(m)=\lcm \{\pi_U(p_1^{v_1}),\ldots,
    \pi_U(p_r^{v_r}),\pi_U(c)\}.    
    \end{equation*}
    Notice that $m$ tends to infinity if and only if at least one of
    the $v_j$'s or $c$ tends to infinity. 
    
    Assume first that for some $j\in\{1,\ldots,r\}$, $v_j\to +\infty$.
    By assumption we have that
    $\lim_{v_j\to+\infty}N_U(p_j^{v_j})=+\infty$ and by
    Remark~\ref{rem:nupiu}, we get that
    $\lim_{v_j\to+\infty}\pi_U(p_j^{v_j})=+\infty$. Therefore 
    $\pi_U(m)$ takes values larger than any constant by considering an
    integer $m$ which is divisible by a sufficiently large power of
    $p_j$. Again using Remark~\ref{rem:nupiu}, the same conclusion
    holds for $N_U(m)$.
    
    Let $C=\{c_0<c_1<c_2<\cdots\}$ be the set of natural numbers prime
    to $a_k$. For all $c\in C$, the sequence $(U_i\mod c)_{i\ge 0}$ is
    ultimately periodic but it is even purely periodic.  Indeed, for
    all $i\ge 0$, $U_{i+k}$ is determined by the $k$ previous terms
    $U_{i+k-1},\ldots,U_i$. But since $\gcd(a_k,c)=1$, $a_k$ is
    invertible modulo $c$ and for all $i\ge 0$, $U_i\mod c$ is also
    determined by the $k$ following terms $U_{i+1},\ldots,U_{i+k}$.
    By definition of $N_U(c)$, the sequence $(U_i\mod c)_{i\ge 0}$
    takes exactly $N_U(c)$ different values because any term appears
    infinitely often. Let $\alpha$ be the function mapping
    $m\in\mathbb{N}$ onto the smallest index $\alpha(m)$ such that
    $U_{\alpha(m)}\ge m$.  Since $U$ is increasing, $\alpha$ is
    non-decreasing and $\lim_{m\to+\infty}\alpha(m)=+\infty$. From
    this last observation and from the pure periodicity of $(U_i\mod
    c)_{i\ge 0}$, it follows that for all $c\in C$,
    $N_U(c)\ge\alpha(c)$ because $U_0<\cdots<U_{\alpha(c)-1}<c$.
    Consequently, we get
    $$\lim_{n\to+\infty}N_U(c_n)=+\infty.$$
    
    Any large enough integer $m$ contains either a large power of some
    $p_j$ or some large $c$ prime to $a_k$ and consequently (i) holds.    
\end{proof}

\begin{theorem}\label{the:finfin}
    Let $U=(U_i)_{i\ge 0}$ be a linear numeration system such that
    $\mathbb{N}$ is $U$-recognizable, satisfying  condition
    \eqref{eq:condlim}. Assume that 
    $$\lim_{m\to+\infty}N_U(m)=+\infty.$$
    Then it is decidable whether
    or not a $U$-recognizable set is ultimately periodic.
\end{theorem}

\begin{proof}
    The sequence $U$ satisfies a recurrence relation of order $k$ of
    the kind \eqref{eq:linrec}.  Let the prime decomposition of $|a_k|$
    be $|a_k|=p_1^{u_1}\cdots p_r^{u_r}$ with $u_1,\ldots,u_r> 0$.
    Consider a DFA $\mathcal{A}$ with $d$ states accepting a
    $U$-recognizable set $X\subseteq\mathbb{N}$. Assume that $X$ is
    ultimately periodic with a period
    $$p_X=p_1^{v_1}\cdots p_r^{v_r}c$$
    where $\gcd(a_k,c)=1$ and
    $v_1,\ldots,v_r\ge 0$.
    
    Since $\gcd(a_k,c)=1$, with the same reasoning as in the proof of
    the previous lemma, the sequence $(U_i\mod c)_{i\ge 0}$ is purely
    periodic. Therefore, $U_0=1$ appears infinitely often in $(U_i\mod
    c)_{i\ge0}$. Since $c$ is a divisor of $p_X$, we can use
    Proposition~\ref{pro:alternative} and get $c\le d$.
    
    By Proposition~\ref{pro:1}, we get $N_U(p_X)\le d$. Let
    $j\in\{1,\ldots,r\}$. Using Remark~\ref{rem:nupiu}, we obtain
    $$
    N_U(p_j^{v_j})\le \pi_U(p_j^{v_j})\le \pi_U(p_X)\le
    (N_U(p_X))^k\le d^k.$$
    The assumption
    $\lim_{m\to+\infty}N_U(m)=+\infty$ implies that
    $\lim_{v\to+\infty}N_U(p_j^v)=+\infty$.  Observe that
    $N_U(p_j^v)\le N_U(p_j^w)$ whenever $v\le w$. Consequently the
    exponent $v_j$ occurring in the decomposition of $p_X$ is bounded
    by $s_j$ where $s_j$ is the smallest integer such that for all
    $v\ge s_j$, $N_U(p_j^v)>d^k$. This bound $s_j$ can be effectively
    computed as follows. For any $v$, $N_U(p_j^v)$ can be found in a
    finite number of operations by inspecting the first values of
    $(U_i\mod p_j^v)_{i\ge 0}$ and looking for two identical
    $k$-tuples made of $k$ consecutive elements.  Once the period is
    determined, one immediately gets the values that are repeated
    infinitely often.  Since the map $v\mapsto N_U(p_j^v)$ is
    non-decreasing, one has to compute $N_U(p_j)\le N_U(p_j^2)\le
    \cdots$ until finding the first value $s_j$ such that
    $N_U(p_j^{s_j})>d^k$.
    
    If $X$ is ultimately periodic, then the admissible periods are
    bounded by the constant
    $$P=p_1^{s_1}\cdots p_r^{s_r}d$$
    which is effectively computable.
    Then, using Proposition~\ref{pro:preper}, the admissible
    preperiods $a_X$ must satisfy $$|\rep_U(a_X-1)|\le d+\max_{p\le
      P}(\iota_U(p))$$
    where $|\rep_U(a)|\le|\rep_U(b)|$ whenever
    $a\le b$. This last observation shows that a bound on the
    admissible preperiods of $X$ can be given effectively.

    Consequently the sets of admissible preperiods and periods that we
    have to check are finite. For each pair $(a,p)$ of admissible
    preperiods and periods, there are at most $2^a2^p$ distinct
    ultimately periodic sets. Thanks to Lemma~\ref{lem:upU}, one can
    build an automaton for each of them and then compare the language
    $L$ accepted by this automaton with $\rep_U(X)$. Recall
    that testing whether $L\setminus\rep_U(X)=\emptyset$ and
    $\rep_U(X)\setminus L=\emptyset$ is decidable
    algorithmically.
\end{proof}

In view of the previous result, it is natural to characterize linear
recurrence sequences $U$ such that $\lim_{m\to+\infty}N_U(m)=+\infty$.
It is exactly the aim of Section~\ref{sec:lrrc}. For instance, we have
the immediate special case that was treated in \cite{CR}.

\begin{corollary}\label{the:the}
    Let $U=(U_i)_{i\ge 0}$ be a linear numeration system such that
    $\mathbb{N}$ is $U$-recognizable and satisfying a recurrence
    relation of order $k$ of the kind \eqref{eq:linrec} with $a_k=\pm
    1$ and condition \eqref{eq:condlim}.
    It is decidable whether or not a $U$-recognizable set is
    ultimately periodic.
\end{corollary}

\begin{proof}
    It is enough to observe that if $a_k=\pm 1$, then
    $\lim_{m\to+\infty}N_U(m)=+\infty$, see Lemma~\ref{lem:equiv}.
\end{proof}

\begin{example}
    Let $\alpha\in (0,1)$ be irrational, with simple continued
    fraction $\alpha=[1,a_1,a,a_3,a,\ldots]$, that is, $a_{2i}=a$
    $(i\ge 1)$, where $a, a_{2i-1}\in\mathbb{N}\setminus\{0\}$ for all
    $i\ge 1$.  The numerators of its even-indexed convergents are
    given by the recurrence $U_{2i}=(aa_{2i-1}+2)U_{2i-2}-U_{2i-4}$
    $(i\ge 1)$ with initial conditions $U_{-2}=1-a$, $U_0=1$. It was
    shown in \cite{Fra} that every positive integer $n$ has a unique
    representation of the form $n=\sum_{i\ge 0} d_{2i}U_{2i}$, where
    the digits satisfy $0\le d_{2i}\le aa_{2i+1}+1$, and the
    additional condition: If for some $0\le k<\ell$ the digits
    $d_{2k}$ and $d_{2\ell}$ attain their maximum values, then there
    exists $j$ satisfying $k<j<\ell$ such that $d_{2j}<aa_{2j+1}$. It
    was recently employed in \cite{SY}, where it was dubbed
    $\ell$-sequences. The special case $a=a_{2i-1}=1$ for all $i\ge 1$
    is an exotic ternary numeration system since then $0\le d_{2i}\le
    2$. The additional condition then states that between any two
    digits $2$ there must be a digit $0$. Also $\alpha=(1+\sqrt{5})/2$
    is the golden ratio and the $U_{2i}$ are the even-indexed
    Fibonacci numbers. This special case was used to investigate
    irregularities of distribution of sequences \cite{CG1},
    \cite{CG2}. In \cite{FK} it was used to compress sparse
    bit-strings, in \cite{Fra2} for answering a combinatorial question
    raised in \cite{BSS}, and in \cite{Fra3} for providing a
    polynomial-time algorithm for combinatorial games. For all these
    numeration systems Problem~\ref{pb} is decidable, since the
    coefficient of $U_{2i-4}$ is $-1$. The same conclusion holds for
    the numeration system based on the denominators of the even
    convergents, also given in \cite{Fra}.
\end{example}

\begin{remark}\label{rem:notint}
    We have thus obtained a decision procedure for our
    Problem~\ref{pb} when $\lim_{m\to+\infty}N_U(m)=+\infty$ and in
    particular when the coefficient $a_k$ occurring in
    \eqref{eq:linrec} is equal to $\pm 1$. On the other hand, whenever
    $\gcd(a_1,\ldots,a_k)=g\ge 2$, for all $n\ge 1$ and for all $i$
    large enough, we have $U_i\equiv 0 \mod g^n$ and the assumption
    about $N_U(m)$ in Theorem~\ref{the:finfin} does not hold.  Indeed,
    the only value taken infinitely often by the sequence $(U_i\mod
    g^n)_{i\ge 0}$ is $0$, so $N_U(m)$ equals $1$ for infinitely many
    values of $m$. Notice in particular, that the same observation can
    be made for the usual $b$-ary system ($b\ge 2$) where the only
    value taken infinitely often by the sequence $(b^i\mod b^n)_{i\ge
      0}$ is $0$, for all $n\ge 1$.
\end{remark}

To conclude this section, we make a small digression. We show how to
use a result of Engstrom about preperiods \cite{Eng} to get some
special linear numeration systems where
$\lim_{m\to+\infty}N_U(m)=+\infty$.  In Engstrom's paper the problem
of finding a general period for a given recurrence relation modulo $m$
for {\em any} initial conditions is considered. Notice that in
\cite{ward}, M.~Ward considers the problem where the initial
conditions are fixed and then the period modulo $m$ has to be
determined.

\begin{theorem}\cite[Theorem 9]{Eng}\label{the:eng}
    Let $U=(U_i)_{i\ge 0}$ be a linear recurrence sequence of order
    $k$ of the kind \eqref{eq:linrec} and $p$ be a prime divisor of
    $a_k$.  If there exists $s(p)<k$ such that
    $a_k,\ldots,a_{k-s(p)+1}\equiv 0\pmod p$ and $a_{k-s(p)}\not\equiv
    0\pmod p$, then $\iota_U(p^v)\le vs(p)$.
\end{theorem}

\begin{remark}\label{rem:exxx}
    Assume that we are dealing with a linear numeration system
    $U=(U_i)_{i\ge 0}$ satisfying \eqref{eq:linrec} and that the
    assumptions of the previous theorem hold for all prime divisors
    $p$ of $a_k$ (which is equivalent to the fact that
    $\gcd(a_1,\ldots,a_k)=1$).  Let $\chi_U(x)$ be the characteristic
    polynomial of $U$ as defined in \eqref{eq:polcar}.  Assume that
    $\beta>1$ is a root of multiplicity $\ell\ge 1$ of $\chi_U(x)$
    satisfying:
\begin{itemize}
  \item for any other root $\gamma\in\mathbb{C}$ of $\chi_U(x)$, $|\gamma|<\beta$,
  \item $\beta<p^{1/s(p)}$ for all prime divisors $p$ of $a_k$.
\end{itemize}
There exists some constant $c$ such that $U_i\sim c\,
i^{\ell-1}\beta^i$.  Let $p$ be a prime divisor of $a_k$ and $j_p(v)$
be the largest index $j$ such that $U_{j}<p^v$. Let $t>s(p)$ be a real
number such that $\beta<p^{1/t}<p^{1/s(p)}$. For $v$ large enough, we
have $U_{\lfloor vt\rfloor}<p^v$.  Consequently, for $v$ large enough,
$j_p(v)\ge \lfloor vt\rfloor$. From the previous theorem, we have
$\iota_U(p^v)\le vs(p)$. Therefore for $v$ large enough,
$U_{\iota_U(p^v)}<\cdots <U_{j_p(v)}$ are the first terms of the
periodic part of $(U_i\mod p^v)_{i\ge 0}$ and $N_U(p^v)\ge \lfloor
vt\rfloor -vs(p)+1$.  This means that for all prime divisors $p$ of
$a_k$, $N_U(p^v)\to+\infty$ as $v\to+\infty$.  Therefore, by
Lemma~\ref{lem:equiv}, $\lim_{m\to+\infty}N_U(m)=+\infty$ and we can
apply our decision procedure given by Theorem~\ref{the:finfin}
(whenever $\mathbb{N}$ is $U$-recognizable).
\end{remark}

\begin{example}
    Consider the linear recurrence sequence given by $U_{i+3}=U_{i+1}+3U_i$ for $i\ge
    0$ and $U_i=i+1$ for $i=0,1,2$. The first terms of the sequence are
    $$1, 2, 3, 5, 9, 14, 24, 41, 66, 113, 189, 311, 528, 878, 1461,
    2462, 4095,\ldots.$$
    With the above notation, $\beta\simeq
    1.6717<3$ and the other two complex roots have modulus close to
    $1.34$. We also have $s(3)=1$.  Thanks to Theorem~\ref{the:eng},
    the preperiod $\iota_U(3^v)$ is bounded by $v$. On the other hand,
    we have $U_i\sim c\, \beta^i$ for some $c>0$. Notice that
    $\beta<3^{1/2}<3$. Therefore for $v$ large enough, $U_{2v}\sim c\,
    \beta^{2v}<3^v$. Consequently, the elements $U_{v}<\cdots<U_{2v}$
    appear in the periodic part. In the following table, these
    elements have been underlined.
$$\begin{array}{l||l|l}
v & \text{preperiod} & \text{period} \\
\hline
3 & 1, 2, 3 & (\underline{5, 9, 14, 24}, 14, 12, 5, 0, 14, 15, 14, 3, 5, 18, 14, 6, 14, 21)\\
4 & 1, 2, 3, 5 & (\underline{9, 14, 24, 41, 66}, 32, 27, 68, 42, 68, 3, 32, 45, 41, 60, 14, \ldots)\\
5 & 1, 2, 3, 5, 9 & (\underline{14, 24, 41, 66, 113, 189}, 68, 42, 149, 3, 32, 207, 41, 60, 
176, \ldots) \\
\end{array}$$
\end{example}

\section{Linear Recurrence Sequences and Residue Classes}\label{sec:lrrc}

As was observed in Remark~\ref{rem:notint}, since our approach to
solve the decision problem is requiring that
$\lim_{m\to+\infty}N_U(m)=+\infty$, it can only be applied to linear
recurrence sequences \eqref{eq:linrec} for which
$\gcd(a_1,\ldots,a_k)=1$. In this section, our aim is to determine
which linear recurrence sequences $U$ are such that
$\lim_{m\to+\infty}N_U(m)=+\infty$. To that end, it is clear (see
Lemma~\ref{lem:equiv}) that we have only to focus on the behavior of
$N_U(p^v)$ for any prime $p$ dividing $a_k$.  \medskip

Throughout this section we let $(U_i)_{i\ge 0}$ be a linear recurrence
sequence satisfying \eqref{eq:linrec}. We assume that $(U_i)_{i\ge 0}$
satisfies no recurrence of smaller order than $k$. It is well-known
(this result is sometimes referred as Kronecker's theorem, see
\cite{kob}) that under the assumption $a_k\neq 0$, this is equivalent
to assume that
$$\det
\begin{pmatrix}
    U_0 &\cdots & U_{k-1}\\
    \vdots& &\vdots\\
    U_{k-1}& \cdots & U_{2k-2}\\
\end{pmatrix}\neq 0.$$
\medskip

Let $p$ be a prime number. We recap some background on the $p$-adic
numbers (see for instance \cite{padic}). We can put an absolute value
$|\cdot|_p$ on $\mathbb{Z}$ as follows.  For each integer $n\neq 0$,
we can write $n=p^v \ell$ with $\ell$ such that $\gcd(p,\ell)=1$.  We
define $|n|_p = p^{-v}$ and $|0|_p=0$.  We note that this absolute
value extends on $\mathbb{Q}$ by declaring $|a/b|_p = |a|_p/|b|_p$ for
$a,b\in\mathbb{Z}$, $b\neq 0$. In particular, for all
$a,b\in\mathbb{Q}$, $|a.b|_p=|a|_p.|b|_p$.  Note that this absolute value is
\emph{non-Archimedian}; that is, it satisfies for all
$a,b\in\mathbb{Q}$
$$|a+b|_p \le \max\{ |a|_p,|b|_p\}.$$
If we complete $\mathbb{Q}$ with
respect to this absolute value, we obtain the field of $p$-\emph{adic rationals},
which we denote by $\mathbb{Q}_p$.  We can take the algebraic closure
of $\mathbb{Q}_p$; the absolute value $|\cdot|_p$ extends to this algebraic
closure.  The algebraic closure is not complete, however.  Completing
this algebraic closure, we obtain a complete algebraically closed
field $\mathbb{C}_p$ with absolute value $|\cdot|_p$, which restricts to the
$p$-adic absolute value on $\mathbb{Q}$. The closed unit ball $$\{x\in
\mathbb{Q}_p:|x|_p\le 1\}$$
is called the set of $p$-\emph{adic integers} and
we denote it by $\mathbb{Z}_p$.  The ordinary integers are dense in
$\mathbb{Z}_p$. The $p$-adic rationals can be viewed as the formal expressions of the form 
$$c_{-N}p^{-N}+\cdots+c_{-1}p^{-1}+c_0+c_1p+c_2p^2+\cdots$$
where
$c_j\in\{0,\ldots,p-1\}$, $N\in\mathbb{Z}$. The $p$-adic integers are
identified with formal expressions involving only non-negative powers
of $p$. \medskip

We let
\begin{equation}
    \label{eq:polcar}
    \chi_U(x)=x^k-a_1x^{k-1}-\cdots-a_k
\end{equation}
be the characteristic
polynomial of $(U_i)_{i\ge 0}$ and
\begin{equation} 
P_U(x)= x^k \chi_U(1/x) = 1-a_1x - \cdots - a_k x^k.
\end{equation}
For $a_k\neq 0$, observe that if $\alpha_1,\ldots,\alpha_s$ are the
roots of $\chi_U$, then the reciprocals $1/\alpha_1,\ldots,1/\alpha_s$
are exactly the roots of $P_U$.
\begin{remark}\label{rem:rat}
    If $(U_i)_{i\in\mathbb{N}}$ satisfies \eqref{eq:linrec} with
    initial conditions $U_0,\ldots,U_{k-1}$, then the
    ordinary power series generating function is a rational function:
    $${\sf U}(x):=\sum_{i\ge 0} U_i\, x^i=\frac{\sum_{i=0}^{k-1} U_i \,
      x^i-\sum_{i+j<k} a_i U_j\, x^{i+j}}{P_U(x)}.$$
    The converse also
    holds. The sequence of Taylor coefficients of any rational
    function $R(x)/P_U(x)$ where $R(x)$ is a polynomial of degree less
    than $k$, satisfies \eqref{eq:linrec}. See for instance \cite[p.
    6]{poorten}.
\end{remark}
Our goal is to prove the following result. Notice that similar
developments can be found in \cite{rauzy}.
\begin{theorem}\label{the:jason} We have $N_U(p^v)\not \rightarrow +\infty$ as $v\rightarrow +\infty$ if and only if
    $P_U(x)=A(x)B(x)$ with $A(x),B(x)\in \mathbb{Z}[x]$ such that:
\begin{enumerate}
\item[(i)] $B(x)\equiv 1\ (\bmod\ p\mathbb{Z}[x])$;
\item[(ii)] $A(x)$ has no repeated roots and all its roots are roots of unity.
\end{enumerate}
In that case, we have moreover $A(0)=B(0)=1$.
\end{theorem}

\begin{proof}
    We note that one direction is fairly simple. Assume that $P_U(x)$ has such
    a factorization $P_U(x)=A(x)B(x)$. From (ii), there is a natural number
    $d$ such that $A(x)|(x^d-1)$.  In view of Remark~\ref{rem:rat},
    there exist some polynomials $Q(x), R(x)$ such that 
    $$(x^d-1)\ \sum_{i\ge0} U_i\, x^i\ = \ \frac{(x^d-1)\, Q(x)}{P_U(x)}\ 
    = \ \frac{(x^d-1)\, Q(x)}{A(x)B(x)} \ = \ \frac{R(x)}{B(x)}.$$
    By
    assumption (i), there exists some integer polynomial $B_1(x)$
    such that $B(x)=1-pB_1(x)$.  Hence
    $$(x^d-1)\sum_{i\ge 0} U_i\, x^i\ = \ \frac{R(x)}{1-pB_1(x)} \ = \ \sum_{i\ge 0}
    p^i\, R(x)\, B_1(x)^i.$$
    In particular, for any fixed integer $v$,
    $(x^d-1)\sum_{i\ge 0} U_i\, x^i$ is congruent to a polynomial (mod
    $p^v$).  This means that $U_{i+d}\equiv U_i~(\bmod~p^v)$ for all
    $i$ sufficiently large.  In particular, there are at most $d$
    values which can occur infinitely often mod $p^v$; that is,
    $N_U(p^v)\le d$ for every $v$.

\vskip .25cm

To do the other direction is a little more work and we use $p$-adic
methods.  We first note that $v\mapsto N_U(p^v)$ is a non-decreasing function,
i.e.,
\begin{equation}\label{eq:nondec}
N_U(p^w) \ge N_U(p^v) \textnormal{ whenever } w\ge v.
\end{equation}
In particular, if $N(p^v)\not\rightarrow+\infty$ then there is some $d$
such that $N_U(p^v)=d$ for all $v$ sufficiently large.  We can pick
integers $a_{1,v},\ldots ,a_{d,v}$ such that if $U_i\equiv a~(\bmod
~p^v)$ for infinitely many $i$, then $a\equiv a_{j,v}~(\bmod ~p^v)$
for some $j\in\{1,\ldots,d\}$.  Since
$$\{(a_{1,w} \bmod p^v),\ldots , (a_{d,w} \bmod p^v)\} = \{(a_{1,v}
\bmod p^v),\ldots , (a_{d,v} \bmod p^v)\}$$
for $w\ge v$, there is no
loss of generality to assume that
$$a_{j,w}\equiv a_{j,v}~(\bmod ~p^v)$$
for $w\ge v$, $1\le j\le d$.
It follows that for $1\le j\le d$ the sequence $(a_{j,v})_{v\ge 1}$
is Cauchy in $\mathbb{Z}_p$.  Thus there exist $b_1,\ldots ,b_d\in
\mathbb{Z}_p$ such that for $1\le j\le d$, 
\begin{equation}\label{eq:vn}
a_{j,v}\rightarrow b_j \text{ as }v\rightarrow+\infty.
\end{equation}
Let 
$$V_i \ = \ \prod_{j=1}^d (U_i-b_j) \in \mathbb{Z}_p.$$
Note that
since the set of linear recurrence sequences is closed under (Hadamard)
product, sum, and scalar multiplication (see for instance
\cite{poorten}), the sequence $(V_i)_{i\ge 0}$ satisfies a linear
recurrence over $\mathbb{Z}_p$.  By construction, $V_i$ is eventually
in $p^v\mathbb{Z}_p$ for any fixed $v$, since any values of $U_i$ that
are not congruent to one of $b_1,\ldots ,b_d$ mod $p^v$ can only occur
finitely many times.  That is for any $v$, $|V_i|_p\le p^{-v}$ for $i$
sufficiently large.  Hence $|V_i|_p\rightarrow 0$ as $i\rightarrow
+\infty$.  Since $(V_i)_{i\ge 0}$ satisfies a linear recurrence, the
power series
$${\sf V}(x) := \sum_{i\ge 0} V_i x^i$$
is a rational power series in
$\mathbb{Q}_p(x)$.  Moreover, ${\sf V}(x)$ converges on the closed
unit disc $\mathbb{Z}_p$, since $|V_i|_p\rightarrow 0$ (in
non-Archimedian fields this is enough to guarantee convergence: a
series $\sum_{i\ge 0} \gamma_i$ converges in $\mathbb{Q}_p$ if and only if
$\lim_{i\to+\infty}|\gamma_i|_p=0$).  Since ${\sf V}(x)$ is a rational
series and it converges on the unit disc, its poles $\beta_1,\ldots
,\beta_r\in \mathbb{C}_p$ must satisfy $|\beta_j|_p>1$, for $1\le j\le
r$.
\medskip

To finish the proof, we will make use of Lemma~\ref{lem:free} stated
below. For the sake of clarity, we have separated this technical
result from this proof.  We note that the statement of this lemma is
very close to what we already have, but it makes the additional
assumption that the poles of ${\sf U}(x)$ generate a free Abelian
subgroup. In general, the poles of ${\sf U}(x)$ generate a finitely
generated Abelian subgroup of $\mathbb{C}_p^{\times}$. From the
so-called {\it fundamental theorem of finitely generated Abelian
  groups} (see for instance \cite[p.  141]{bas}), this group is
isomorphic to $\mathbb{Z}^e \times T$, for some finite Abelian group
$T$ and integer $e\ge 0$.

Let us show how to get rid of the torsion group $T$ to be able to
invoke Lemma~\ref{lem:free}.  Let $a=\# T$. For $0\le b<a$, instead of
taking the sequence $(U_i)_{i\ge 0}$, consider the sequence
$(U^{(b)}_i)_{i\ge 0}:=(U_{ai+b})_{i\ge 0}$.  This latter sequence
satisfies a linear recurrence and the poles of the generating function
${\sf U}^{(b)}(x)$ of $(U^{(b)}_i)_{i\ge 0}$ are the $a$'th powers\footnote{If we
  consider the exponential sum $U_i=\sum_{j=1}^s q_j(i)\, \alpha_j^i$
  where the $q_j$'s are polynomials in $\mathbb{C}_p[x]$ and the
  $\alpha_j$'s are the roots of the characteristic polynomial of
  $(U_i)_{i\ge 0}$, i.e., the reciprocals of the poles of ${\sf U}(x)$,
  then we obtain that $U_i^{(b)}=U_{ai+b}=\sum_{j=1}^s
  \alpha_j^bP_j(ai+b)\, (\alpha_j^a)^i$. Hence the poles of ${\sf
    U}^{(b)}(x)$ are the $a$'th powers of the poles of ${\sf U}(x)$.}
of the poles of ${\sf U}(x)$.  Consequently, the poles of ${\sf
  U}^{(b)}(x)$ generate a finitely generated torsion-free Abelian
group, which is necessarily a free Abelian group.

Since the poles $\beta_1,\ldots ,\beta_r\in \mathbb{C}_p$ of ${\sf
  V}(x)$ satisfy $|\beta_j|_p>1$, for $1\le j\le r$, with the same
reasoning, we get that the poles of the rational function
$${\sf V}^{(b)}(x)=\sum_{i\ge 0} V_{ai+b}\, x^i$$
are $a$'th powers of
$\beta_1,\ldots ,\beta_r$ and for $1\le j\le r$,
$|\beta_j^a|_p=|\beta_j|_p^a>1$. We can invoke Lemma \ref{lem:free}
applied to the sequence $(U^{(b)}_i)_{i\ge 0}$ (it satisfies some
linear recurrence sequence with integer coefficients), and we deduce
that any pole $\gamma$ of ${\sf U}^{(b)}(x)$ satisfies either
$|\gamma|_p>1$ or $\gamma=1$.  In particular, the distinct poles
$1/\alpha_1,\ldots,1/\alpha_s$ of ${\sf U}(x)$, with $\alpha_1,\ldots
,\alpha_s\in \mathbb{C}_p$, being $a$'th roots of the poles of ${\sf
  U}^{(b)}(x)$, are either such that $|1/\alpha_j|_p>1$ or roots of
unity.  By minimality of the order $k$ of the recurrence satisfied by
$(U_i)_{i\ge 0}$, the poles of ${\sf U}(x)$ are precisely the roots of
$P_U(x)$.  We factor
$$P_U(x)=(1-\delta_1 x)\cdots (1-\delta_k x)$$
(each $\delta_j$ is one
of $\alpha_1,\ldots ,\alpha_s$, although they may be repeated). Let us
factor $P_U(x)$ as $A(x)B(x)$ where
$$A(x)=\prod_{\{j:|\delta_j|_p=1\}} (1-\delta_j x)\quad \text{ and
}\quad B(x)=\prod_{\{j:|\delta_j|_p<1\}} (1-\delta_j x).$$
By
assumption $P_U(x)\in \mathbb{Z}[x]$. Moreover if $K$ is a splitting
field of $P_U(x)$ over $\mathbb{Q}$ then any automorphism of $K$ must
permute the set of $\delta_j$ with $|\delta_j|_p<1$, since the
automorphism permutes the entire set of $\delta_j$'s and it must send
roots of unity to roots of unity.  Thus $B(x)$ is a rational
polynomial, since it is fixed by every automorphism of $K$.
 Note that if $n>0$, then the coefficient of
$x^n$ in $B(x)$ is given by a sum of products of $n$ elements in
$\{\delta_j:|\delta_j|_p<1\}$.  The set of algebraic integers is a
subring of $\mathbb{C}_p$ and the only rationals that are algebraic
integers are in fact integers.  Since the $\delta_j$'s are algebraic
integers, $B(x)$ is thus an integer polynomial.  Moreover, since the
$p$-adic absolute value is non-Archimedian, the coefficient of $x^n$
in $B(x)$, $n>0$, has $p$-adic absolute value strictly less than $1$.
Note that an integer $m$ satisfying $|m|_p<1$ is necessarily a
multiple of $p$.  Hence $B(x)\equiv 1\bmod p\mathbb{Z}[x]$.

Let us turn to the polynomial $A(x)$. The roots
of $A(x)$ are roots of unity.  Moreover, $A(x)\in\mathbb{Z}[x]$ by the
same reasoning as before.

To finish the proof, we have one last thing to show: we need to know
that $A(x)$ has no repeated roots.  To do this, we show that the poles
of ${\sf U}(x)$ that are roots of unity are simple.  Recall that
$1/\alpha_1,\ldots,1/\alpha_s$ are the distinct poles of ${\sf U}(x)$.
We may assume that there exists $t\ge 0$ such that $\alpha_1,\ldots
,\alpha_t$ are roots of unity and that for $j=t+1,\ldots,s$,
$|\alpha_j|_p<1$.  Then there exist polynomials
$q_j\in\mathbb{C}_p[x]$, $j=1,\ldots,s$, such that we have for all
$i$,
$$U_i \ = \ \sum_{j=1}^s q_j(i)\, \alpha_j^i=\underbrace{\sum_{j=1}^t
  q_j(i)\, \alpha_j^i}_{:=T_i} +\underbrace{\sum_{j=t+1}^s q_j(i)\,
  \alpha_j^i}_{:=W_i}.$$
Since for $j=t+1,\ldots,s$, $|\alpha_j|_p<1$,
we get that $|U_i-T_i|_p=|W_i|_p\rightarrow 0$ as
$i\rightarrow+\infty$. Indeed, for any $j$,
$\{|q_j(i)|_p:i\in\mathbb{N}\}$ is bounded by a constant.  Since for
$1\le j\le t$, $\alpha_j$ is a root of unity, there exists a natural
number $a$ such that $\alpha_j^a=1$ for all $j\in\{1,\ldots,t\}$.  As
before, we let $T^{(b)}_i = T_{ai+b}$ for $0\le b<a$. Thus $$T^{(b)}_i
\ = \ \sum_{j=1}^t q_j(ai+b)\, \alpha_j^{ai+b}\ = \ \sum_{j=1}^t
\alpha_j^{b}\, q_j(ai+b)$$
is a polynomial with coefficients in
$\mathbb{C}_p$ denoted by $g_b(i)$.

Let $\epsilon>0$. By definition of the $b_j$'s (see \eqref{eq:vn}
above) for all large enough $i$, there exists
$\ell(i)\in\{1,\ldots,d\}$ such that $|U_i-b_{\ell(i)}|_p<\epsilon$.
Since $|U_i-T_i|_p\rightarrow 0$ as $i\to+\infty$, we get
$$|(T_i-b_1)\cdots (T_i-b_d)|_p \le \prod_{j=1}^d
(|T_i-U_i|_p+|U_i-b_{\ell(i)}|_p+|b_{\ell(i)}-b_j|_p) \rightarrow 0$$
as
$i\rightarrow+\infty$ because every factor is bounded by a constant and
one tends to zero.  Thus for $0\le b<a$, as $T_i^{(b)}=g_b(i)$, we have 
$$|(g_b(i)-b_1)\cdots (g_b(i)-b_d)|_p\rightarrow 0 \text{ as
}i\rightarrow +\infty.$$
Consider the polynomial defined by
$$h_b(i):=(g_b(i)-b_1)\cdots (g_b(i)-b_d).$$
Let $n_0$
be a natural number.  Notice that by the Binomial Theorem, we get
$|h_b(n_0+p^v ) -h_b(n_0)|_p\le  C p^{-v}\to 0$ as $v\rightarrow+\infty$, where $C$ is the maximum of the $p$-adic absolute values of the coefficients of $h_b$. So since
$|h_b(i)|_p\rightarrow 0$ as $i\rightarrow+\infty$, we
get\footnote{Indeed, $|h_b(n_0)|_p\le |h_b(n_0)-h_b(n_0+p^v
  )|_p+|h_b(n_0+p^v)|_p$ and both terms tend to $0$ as $v\to+\infty$.}
that $h_b(n_0)=0$.  Thus for all integer $i$ each $h_b(i)=0$, $0\le
b<a$, and so each $g_b$ is a constant polynomial (this is a
consequence of the {\em Fundamental Theorem of Algebra}, $h_b(i)=0$
for all $i$ implies that $g_b$ takes infinitely often the same value
amongst $b_1,\ldots,b_d$ so it is constant).  It follows that $T_i$ is
periodic:
  $$\forall i\ge 0,\ T_{i+a}=T_i.$$
  So $(x^a-1){\sf U}(x)$ has no poles on
  the closed unit disc so in particular on the unit circle, as
  $$|U_{i+a}-U_{i}|_p\le
  |T_{i+a}-T_{i}|_p+|W_{i+a}-W_{i}|_p\rightarrow 0$$
  as
  $i\rightarrow+\infty$.  We know that ${\sf U}(x)=Q(x)/P_U(x)$ with
  $Q$ and $P_U$ relatively prime (by minimality assumption on order
  of recurrence satisfied by $(U_i)_{i\ge 0}$).  Consequently, $A(x)$
  divides $(x^a-1)$ because $(x^a-1)Q(x)/P_U(x)$ has no poles on the
  unit circle.  This shows that it has no repeated roots, completing
  the proof.
\end{proof}

Let us consider the technical lemma used in the previous proof.
\begin{lemma}\label{lem:free}
    Let $(U_i)_{i\ge 0}$ be an integer linear recurrence sequence
    satisfying \eqref{eq:linrec} and let ${\sf U}(x) \ = \ \sum_{i\ge
      0} U_i x^i \in \mathbb{Z}_p[[x]]$ be the corresponding rational
    power series.  Suppose that the multiplicative subgroup of
    $\mathbb{C}_p^{\times}$ generated by the (finitely many) poles of
    ${\sf U}(x)$ is a free Abelian group and let $V_i = \prod_{j=1}^d
    (U_i-b_j)$ for $i\ge 0$ where $b_1,\ldots ,b_d\in \mathbb{Z}_p$.
    If the rational power series ${\sf V}(x) \ = \ \sum_{i\ge 0} V_i
    x^i \in \mathbb{Z}_p[[x]]$ has poles $\beta_1,\ldots ,\beta_r\in
    \mathbb{C}_p$ satisfying $|\beta_j|_p>1$ for $1\le j\le r$, then
    every pole $\gamma\in \mathbb{C}_p$ of ${\sf U}(x)$ either
    satisfies $|\gamma|_p>1$ or $\gamma=1$.
\end{lemma}

\begin{proof}
    Let $1/\alpha_1,\ldots,1/\alpha_s$ be the distinct poles of ${\sf
      U}(x)$ with $\alpha_1,\ldots,\alpha_s\in\mathbb{C}_p$. Note that
    $0$ cannot be a pole of ${\sf U}$.  We first claim that
    $|\alpha_j|_p\le 1$ for $1\le j\le s$.  To see this, notice that
    to be a pole, each $1/\alpha_j$ must satisfy $P_U(1/\alpha_j)=0$.
    This means that $$1-a_1/\alpha_i - \cdots -a_k/\alpha_i^k=0.$$
    Consequently, for all $j\in\{1,\ldots,s\}$, we have
    $$\left|a_1/\alpha_j + \cdots + a_k/\alpha_j^k\right|_p =
    |1|_p=1.$$
    By the non-Archimedian property we have
    $|a_\ell/\alpha_j^\ell|_p\ge 1$ for some $\ell\in\{1,\ldots,k\}$.
    Thus $|\alpha_j|_p^\ell \le |a_\ell|_p$.  Since $a_\ell\in
    \mathbb{Z}$, $|a_\ell|_p\le 1$.  This gives the claim and so
    $|\alpha_j|_p\le 1$ for $1\le j\le s$.  We may assume that
    $|\alpha_1|_p=\cdots = |\alpha_t|_p=1$ and $|\alpha_j|_p<1$ for
    $0\le t<j\le s$. So, we know that the poles of ${\sf U}$ have $p$-adic
    absolute value $\ge 1$.
    
    Next there exist polynomials $q_1(x),\ldots ,q_s(x)\in
    \mathbb{C}_p[x]$ such that
    $$U_i \ = \ \sum_{j=1}^s q_j(i)\alpha_j^i.$$
    Moreover, define 
    $$V_i:=(U_i-b_1)\cdots (U_i-b_d) = c_d U_i^d + c_{d-1}
    U_i^{d-1}+\cdots + c_0$$
    for some $c_0,\ldots ,c_{d-1}, c_d=1\in
    \mathbb{Z}_p$.  Hence by Multinomial Theorem, we have
\begin{eqnarray*}
V_i &=& \sum_{j=0}^d c_j \sum_{\substack{j_1+\cdots +j_s=j\\ j_1,\ldots,j_s\ge 0}} {j\choose j_1,\ldots , j_s}\prod_{\ell=1}^s (q_{\ell}(i)\alpha_{\ell}^i)^{j_{\ell}}\\
&=&  \sum_{j=0}^d c_j \sum_{\substack{j_1+\cdots +j_s=j\\ j_1,\ldots,j_s\ge 0}} {j\choose j_1,\ldots , j_s}\prod_{\ell=1}^s q_{\ell}(i)^{j_{\ell}}\left( \prod_{\ell=1}^s \alpha_{\ell}^{j_{\ell}}\right)^i.
\end{eqnarray*}
Since the roots of the characteristic polynomial are the reciprocals
of the poles of the corresponding rational power series, it follows
that the set of poles of ${\sf V}(x)$ is contained in the set
$$\left\{ \prod_{\ell=1}^s
    \alpha_{\ell}^{-j_{\ell}}~:~j_1,\ldots,j_s\ge 0,\ j_1+\cdots
    +j_s\le d\right\}.$$
By assumption, the poles of ${\sf V}(x)$ all
have $p$-adic absolute value strictly greater than $1$.  Note that
$\left|\prod_{\ell=1}^r \alpha_{\ell}^{-j_{\ell}}\right|_p =
\left|\prod_{\ell=1}^t \alpha_{\ell}^{-j_{\ell}}\right|_p
\left|\prod_{\ell=t+1}^r \alpha_{\ell}^{-j_{\ell}}\right|_p$ is $>1$
if and only if $j_{\ell}>0$ for some $\ell>t$.  Therefore, we can
conclude that the possible poles of ${\sf V}$ supported only by
products on $\alpha_1,\ldots,\alpha_t$ do not occur.  So for instance,
$\alpha_1^d$ is not a pole of ${\sf V}$.

Let $G$ denote the multiplicative subgroup of $\mathbb{C}_p^{\times}$
generated by $\alpha_1,\ldots ,\alpha_t$.  By assumption, $G$ is a
subgroup of a finitely generated free Abelian group, and hence $G\cong
\mathbb{Z}^{e}$, for some natural number $e\ge 0$. To conclude the
proof, it is enough to show that $e=0$. Because in that case, one can
conclude that the only possible pole $\gamma$ of ${\sf U}$ such that
$|\gamma|_p=1$ is $1$. To see this, suppose that $e>0$ and let
$\gamma_1,\ldots ,\gamma_e$ be generators for a free Abelian group of
rank $e$. Then for $1\le i\le t$, we can write
$$\alpha_i = \prod_{j=1}^e \gamma_j^{b_{i,j}},$$
where $b_{i,j}$ are
integers.  We relabel if necessary so that
 \begin{itemize}
 \item 
 $|b_{1,1}|=\max\{|b_{i,1}|\}>0$; 
 \item $|b_{1,j}|=\max\{|b_{i,j}|~:~b_{i,\ell}=b_{1,\ell} \textnormal{ for }\ell<j\}$.
 \end{itemize}
 By construction, $\alpha_1^d$ cannot be written as a different word
 in $\alpha_1,\ldots ,\alpha_t$ of length at most $d$.  Then the
 expression for $V_i$ above has an occurrence of
 $$c_d q_1(i)^{d} \alpha_1^{di}$$
 that cannot be canceled by any
 other pole, by our selection of $\alpha_1$.  Consequently,
 $\alpha_1^d$ should be a pole of ${\sf V}$ which contradicts the
 conclusion of the previous paragraph. So, $e=0$ and the result
 follows.
\end{proof}

As it was stressed in the introduction of this section, we are now
able to check, whether or not the number of values taken infinitely
often by a linear recurrence sequence modulo $p^v$ tends to infinity
as $v\to+\infty$. Let us consider the following example.

\begin{example}\label{exa:aaa}
    Consider the linear recurrence sequence $(U_i)_{i\ge 0}$ given in
    \cite{FroSeq} and defined by $U_{i+4}=3U_{i+3}+2U_{i+2}+3U_{i}$
    for $i\ge 0$ and $U_i=i+1$ for $i=0,\ldots,3$. As shown in
    \cite{FroSeq}, addition within this linear numeration system is
    not computable by a finite automaton. Nevertheless, we can show by
    applying the previous theorem that $N_U(3^v)\to+\infty$ as
    $v\to+\infty$. One has
    $$P_U(x)=1-3x-2x^2-3x^4$$
    and it is not difficult to see that this
    polynomial cannot be factorized as $A(x)B(x)$ with two factors
    satisfying hypotheses of Theorem~\ref{the:jason}. This example
    shows that our decision procedure given by
    Theorem~\ref{the:finfin} can take care of numeration systems not
    handled by \cite{ARS,Leroux,much}.
\end{example}

\begin{example}
    Consider the recurrence relation
    $$U_{i+5} = 6 U_{i +4} + 3U_{i+3} - U_{i+2} + 6U_{i+1} +
    3U_{i},\ \forall i\ge 0.$$
    With the above notation,
    $\chi_U(x)=x^5-6x^4-3x^3+x^2-6x-3$ and
    $$P_U(x)=1-6x-3x^2+x^3-6x^4-3x^5=\underbrace{(x^3 +
      1)}_{A(x)}\underbrace{(-3x^2 - 6x + 1)}_{B(x)}.$$
    With the
    initial conditions $U_i=i+1$ for $i=0,\ldots,4$, the corresponding
    sequence does not satisfy any relation of shorter length as
    $$\det
\begin{pmatrix}
1& 2& 3& 4& 5\\ 
2& 3& 4& 5& 54\\
3& 4& 5& 54& 359\\
4& 5& 54& 359& 2344\\
5& 54& 359& 2344& 15129\\
\end{pmatrix}=8458240\neq 0.$$
Even if the gcd of the coefficients of the recurrence is $1$, since
$P_U$ satisfies the assumptions of Theorem~\ref{the:jason} for $p=3$,
$N_U(3^v)\not\to+\infty$ as $v\to+\infty$. The following table gives
the first values of $N_U(3^v)$.
$$\begin{array}{l||l|c}
 v & \text{period} & N_U(3^v) \\
\hline
1 &   (1, 0, 1, 2, 0, 2) & 3 \\
2 & (4, 0, 1, 5, 0, 8) & 5 \\
3 & (22, 9, 19, 5, 18, 8) & 6\\
4 & (49, 63, 19, 32, 18, 62) & 6\\
5 & (211, 225, 19, 32, 18, 224) & 6\\
\vdots & & \vdots\\
\end{array}$$
\end{example}

\section{A Decision Procedure for a Class of Abstract Numeration Systems}\label{sec:last}

An {\em abstract numeration system} $S=(L,\Sigma,<)$ is given by an
infinite regular language $L$ over a totally ordered alphabet
$(\Sigma,<)$ \cite{LR}. By enumerating the words of $L$ in
genealogical order, we get a one-to-one correspondence denoted
$\rep_S$ between $\mathbb{N}$ and $L$. In particular, $0$ is
represented by the first word in $L$. The reciprocal map associating a
word $w\in L$ to its index in the genealogically ordered language $L$
is denoted $\val_S$. A set $X\subseteq\mathbb{N}$ of integers is {\em
  $S$-recognizable} if the language $\rep_S(X)$ over $\Sigma$ is
regular.  \smallskip

Let $S=(L,\Sigma,<)$ be an abstract numeration system built over an
infinite regular language $L$ having
$\mathcal{M}_L=(Q_L,q_{0,L},\Sigma,\delta_L,F_L)$ as minimal
automaton.  The transition function $\delta_L:Q_L\times\Sigma\to Q_L$
is extended on $Q_L\times \Sigma^*$ by $\delta_L(q,\varepsilon)=q$ and
$\delta_L(q,aw)=\delta_L(\delta_L(q,a),w)$ for all $q\in Q_L$,
$a\in\Sigma$ and $w\in\Sigma^*$. We denote by $\mathbf{u}_i(q)$ (resp.
$\mathbf{v}_i(q)$) the number of words of length $i$ (resp.  $\le i$)
accepted from $q\in Q_L$ in $\mathcal{M}_L$. By classical arguments,
the sequences $(\mathbf{u}_i(q))_{i\ge 0}$ (resp.
$(\mathbf{v}_i(q))_{i\ge 0}$) satisfy the same homogenous linear
recurrence relation for all $q\in Q_L$ (for details, see
Remark~\ref{rem:ujvj}).

In this section, we consider, with some extra hypothesis on the abstract
numeration system, the following decidability question analogous to
Problem~\ref{pb}.

\begin{problem}
    Given an abstract numeration system $S$ and a set
    $X\subseteq\mathbb{N}$ such that $\rep_S(X)$ is recognized by a
    (deterministic) finite automaton, is it decidable whether or not
    $X$ is ultimately periodic, i.e., whether or not $X$ is a finite
    union of arithmetic progressions ?
\end{problem}

Abstract numeration systems are a generalization of positional
numeration systems $U=(U_i)_{i\ge 0}$ for which $\mathbb{N}$ is
$U$-recognizable.

\begin{example}
    Take the language $L=\{\varepsilon\}\cup 1 \{0,01\}^*$ and assume
    $0<1$. Ordering the words of $L$ in genealogical order:
    $\varepsilon, 1,10,100,101,1000,1001,\ldots$ gives back the
    Fibonacci system.
\end{example}

The following example shows that the class of abstract
numeration systems is strictly bigger than the class of linear
numeration systems for which $\mathbb{N}$ is recognizable. 

\begin{example}\label{ex:2fib}
    Consider the language $L=\{\varepsilon\}\cup \{a,ab\}^* \cup \{
      c,cd \}^*$ and the ordering $a<b<c<d$ of the alphabet. If we
      order the first words in $L$ we get
$$\begin{array}{|r|r||r|r||r|r||r|r||r|r|}
\hline
0 & \varepsilon   & 5 & cc  & 10 & ccc & 15 & aaba & 20 & ccdc \\
1 & a  & 6 & cd  & 11 & ccd  & 16 & abaa & 21 & cdcc \\
2 & c  & 7 & aaa & 12 & cdc  & 17 & abab & 22 & cdcd\\
3 & aa & 8 & aab & 13 & aaaa   & 18 & cccc & 23 & aaaaa\\
4 & ab & 9 & aba & 14 & aaab  & 19 & cccd & 24 & aaaab\\
\hline
\end{array}$$

Notice that there is no bijection $\nu:\{a,b,c,d\}\to\mathbb{N}$
between $\{a,b,c,d\}$ and a set of integers leading to a positional
linear numeration system. Otherwise stated, $a,b,c,d$ cannot be
identified with usual ``digits''. Indeed, assume that there exists a
sequence $U$ of integers such that for all $x_1\cdots x_n\in L$, with
$x_i$ in $\{a,b,c,d\}$ for all $i$, $\val_U(\nu(x_1)\cdots
\nu(x_n))=\val_S(x_1\cdots x_n)$. Since $\val_S(a)=1$ and
$\val_S(c)=2$ we get $U_0=1$, $\nu(a)=1$ and $\nu(c)=2$. Moreover,
$\val_S(aa)=3=\nu(a)U_1+\nu(a)U_0$, so $U_1=2$.  Therefore,
$\val_U(\nu(c)\nu(c))=2U_1+2U_0=6$ but $\val_S(cc)=5$ getting a
contradiction. 

\begin{figure}[htbp]
    \centering
    \includegraphics{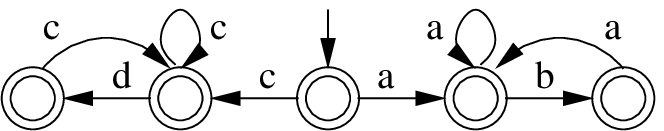}
    \caption{A DFA accepting $L$.}
    \label{fig:01}
\end{figure}
For all $i\ge 1$, we have $\mathbf{u}_i(q_{0,L})=2F_i$ and
$\mathbf{u}_0(q_{0,L})=1$. Consequently, for $i\ge 1$,
$$\mathbf{v}_i(q_{0,L})=1+\sum_{n=1}^i \mathbf{u}_n(q_{0,L})=1+2\sum_{n=1}^i F_n.$$
Notice that for $i\ge 1$,
$\mathbf{v}_i(q_{0,L})-\mathbf{v}_{i-1}(q_{0,L})=\mathbf{u}_i(q_{0,L})=2F_i$. Consequently, by definition of the Fibonacci sequence, we get for all 
$i\ge 3$,
$$\mathbf{v}_i(q_{0,L})-\mathbf{v}_{i-1}(q_{0,L})=(\mathbf{v}_{i-1}(q_{0,L})-\mathbf{v}_{i-2}(q_{0,L}))+
(\mathbf{v}_{i-2}(q_{0,L})-\mathbf{v}_{i-3}(q_{0,L}))$$
and
$$\mathbf{v}_i(q_{0,L})=2\mathbf{v}_{i-1}(q_{0,L})-\mathbf{v}_{i-3}(q_{0,L}), 
\text{ with }\mathbf{v}_0(q_{0,L})=1,\mathbf{v}_1(q_{0,L})=3,\mathbf{v}_2(q_{0,L})=7.$$
\end{example}

\begin{remark}\label{rem:ujvj}
    The computation given in the previous example to obtain a
    homogenous linear recurrence relation for the sequence
    $(\mathbf{v}_i(q_{0,L}))_{i\ge 0}$ can be carried on in general.
    Let $q\in Q_L$. The sequence $(\mathbf{u}_i(q))_{i\ge 0}$
    satisfies a homogenous linear recurrence relation of order $k$
    whose characteristic polynomial is the characteristic polynomial
    of the adjacency matrix of $\mathcal{M}_L$. There exist
    $a_1,\ldots,a_k\in\mathbb{Z}$ such that for all $i\ge 0$,
$$\mathbf{u}_{i+k}(q)=a_1 \mathbf{u}_{i+k-1}(q)+\cdots +a_k \mathbf{u}_i(q).$$
Consequently, we have for all $i\ge 0$
$$\mathbf{v}_{i+k+1}(q)-\mathbf{v}_{i+k}(q)= \mathbf{u}_{i+k+1}(q)=a_1 (\mathbf{v}_{i+k}(q)-\mathbf{v}_{i+k-1}(q))+\cdots +a_k (\mathbf{v}_{i+1}(q)-\mathbf{v}_{i}(q)).$$
Therefore the sequence $(\mathbf{v}_i(q))_{i\ge 0}$ satisfies a
homogenous linear recurrence relation of order $k+1$.
\end{remark}

As shown by the following lemma, in an abstract numeration system, the
different sequences $(\mathbf{u}_{i}(q))_{i\ge 0}$, for $q\in Q_L$,
are replacing the single sequence $(U_i)_{i\ge 0}$ defining a
positional numeration system as in Definition~\ref{def:numsys}.

\begin{lemma}\cite{LR}
Let
$w=\sigma_1\cdots \sigma_n\in L$. We have
\begin{equation}
    \label{eq:lelemme}
    \val_S(w)=\sum_{q\in Q_L}\sum_{i=1}^{|w|}\beta_{q,i}(w)\, \mathbf{u}_{|w|-i}(q)
\end{equation}
where 
\begin{equation}
    \label{eq:beta}
    \beta_{q,i}(w):=\# \{\sigma<\sigma_i\mid
    \delta_L(q_{0,L},\sigma_1\cdots \sigma_{i-1}\sigma) =q\}+\mathbf{1}_{q,q_{0,L}}
\end{equation}
for $i=1,\ldots ,|w|$.
\end{lemma}

Recall that $\mathbf{1}_{q,q'}$ is equal to $1$ if $q=q'$ and it is
equal to $0$ otherwise.

\begin{proposition}\label{pro:pa}\cite{LR}
    Let $S=(L,\Sigma,<)$ be an abstract numeration system built over
    an infinite regular language $L$ over $\Sigma$. Any ultimately periodic set $X$
    is $S$-recognizable and a DFA accepting $\rep_S(X)$ can be
    effectively obtained.
\end{proposition}

Recall that an automaton is {\it trim} if it is accessible and
coaccessible (each state can be reached from the initial state and
from each state, one can reach a final state).

\begin{proposition}\label{pro:absnum}
    Let $S=(L,\Sigma,<)$ be an abstract numeration system such that for
    all states $q$ of the trim minimal automaton
    $\mathcal{M}_L=(Q_L,q_{0,L},\Sigma,\delta_L,F_L)$ of $L$,
    $$\lim_{i\to+\infty}\mathbf{u}_i(q)=+\infty$$
    and $\mathbf{u}_i(q_{0,L})>0$ for all $i\ge 0$. If
    $X\subseteq\mathbb{N}$ is an ultimately periodic set of period
    $p_X$, then any deterministic finite automaton
    accepting $\rep_S(X)$ has at least $\lceil
    N_{\mathbf{v}}(p_X)/\#Q_L\rceil$ states where
    $\mathbf{v}=(\mathbf{v}_i(q_{0,L}))_{i\ge 0}$.
\end{proposition}

\begin{proof} Let $a_X$ be the preperiod of $X$.
    Since for all states $q$ of $\mathcal{M}_L$, we have
    $\lim_{i\to+\infty}\mathbf{u}_i(q)=+\infty$, there exists a minimal
    constant $J>0$ such that $\mathbf{u}_J(q)\ge p_X$ for all $q\in
    Q_L$. Consider for any $i\ge 0$, the word
    $$w_i=\rep_S(\mathbf{v}_i(q_{0,L})),$$ 
    corresponding to the first word of length $i+1$ in the
    genealogically ordered language $L$. Consequently, for $i\ge J-1$,
    $w_i$ is factorized as $w_i=a_i b_i$ with $|b_i|=J$ and we define
    $q_i:=\delta_L(q_{0,L},a_i)$. Notice that $b_i$ is the smallest
    word of length $J$ accepted from $q_i$. By definition of $J$, from
    each $q_i$, there are at least $p_X$ words of length $J$ leading
    to a final state. If we order them by genealogical ordering, we
    denote the $p_X$ first of them by
                $$b_i=b_{i,0}<b_{i,1}<\cdots <b_{i,p_X-1}.$$
                Notice that for $t\in\{0,\ldots,p_X-1\}$, we have
                $$\val_S(a_ib_{i,t})=\val_S(a_ib_i)+t=\mathbf{v}_i(q_{0,L})+t.$$

                The sequence $(\mathbf{v}_i(q_{0,L})\mod p_X)_{i\ge
                  0}$ is ultimately periodic and takes infinitely
                often $N_{\mathbf{v}}(p_X)=:N$ different values. Let
                $h_1,\ldots,h_{N}\ge J-1$ such that
                $$i\neq j\Rightarrow \mathbf{v}_{h_i}(q_{0,L})\not\equiv  \mathbf{v}_{h_j}(q_{0,L})\mod p_X$$
                and for all $i\in\{1,\ldots,N\}$, $\mathbf{v}_{h_i}(q_{0,L})\ge a_X$. We have
                $$\rep_S(\mathbf{v}_{h_i}(q_{0,L}))=w_{h_i}=a_{h_i}b_{h_i}\text{ and }q_{h_i}=\delta_L(q_{0,L},a_{h_i}).$$
                The elements in the set $\{q_{h_1},\ldots,q_{h_N}\}$
                can take only $\# Q_L$ different values. So at least
                $\sigma:=\lceil N/\# Q_L\rceil$ of them are the same.
                For the sake of simplicity, assume that they are
                $q_{h_1},\ldots,q_{h_\sigma}$. Consequently, for
                $i,j\in\{1,\ldots,\sigma\}$ and for all
                $t=0,\ldots, p_X-1$, we have $b_{h_i,t}=b_{h_j,t}$.
                For all $i,j\in\{1,\ldots,\sigma\}$ such that $i\neq
                j$, by Lemma~\ref{lem:per}, there exists $t_{i,j}<p_X$
                such that either $\mathbf{v}_{h_i}(q_{0,L})+t_{i,j}\in
                X$ and $\mathbf{v}_{h_j}(q_{0,L})+t_{i,j}\not\in X$ or,
                $\mathbf{v}_{h_i}(q_{0,L})+t_{i,j}\not\in X$ and
                $\mathbf{v}_{h_j}(q_{0,L})+t_{i,j}\in X$. Therefore,
                the words $a_{h_i}$ and $a_{h_j}$ do not belong to the
                same equivalence class for the relation
                $\sim_{\rep_S(X)}$. This can be shown by concatenating
                the word $b_{h_i,t_{i,j}}=b_{h_j,t_{i,j}}$.  Hence the
                minimal automaton of $\rep_S(X)$ has at least $\sigma$
                states.
\end{proof}

\begin{corollary}\label{cor:ans}
    Let $S=(L,\Sigma,<)$ be an abstract numeration system having the
    same properties as in Proposition~\ref{pro:absnum}. 
                Assume that the sequence $\mathbf{v}=(\mathbf{v}_i(q_{0,L}))_{i\ge
    0}$ is such that 
    $$\lim_{m\to+\infty}N_{\mathbf{v}}(m)=+\infty.$$ Then the period of an
    ultimately periodic set $X\subseteq\mathbb{N}$ such that
    $\rep_S(X)$ is accepted by a DFA with $d$ states is bounded by the
    smallest integer $s_0$ such that for all $m\ge s_0$,
    $N_{\mathbf{v}}(m)>d\,\#Q_L$, where $Q_L$ is the set of states of
    the (trim) minimal automaton of~$L$.
\end{corollary}

\begin{proposition}\label{pro:ans}
    Let $S=(L,\Sigma,<)$ be an abstract numeration system. If
    $X\subseteq\mathbb{N}$ is an ultimately periodic set of period
    $p_X$ such that $\rep_S(X)$ is accepted by a DFA with $d$ states,
    then the preperiod $a_X$ of $X$ is bounded by an effectively
    computable constant $C$ depending only on $d$ and $p_X$.
\end{proposition}

\begin{proof}
    Let $\mathcal{A}=(Q,q_0,\Sigma,\delta,F)$ be a DFA with $d$ states
    accepting $\rep_S(X)$.  As usual,
    $\mathcal{M}_L=(Q_L,q_{0,L},\Sigma,\delta_L,F_L)$ is the minimal
    automaton of $L$ and for any state $q\in Q_L$, $\mathbf{u}_i(q)$
    is the number of words of length $i$ accepted from $q$ in
    $\mathcal{M}_L$. Since $(\mathbf{u}_i(q))_{i\ge 0}$ satisfies a
    linear recurrence relation, the sequences $(\mathbf{u}_i(q)\mod
    p_X)_{i\ge 0}$ are ultimately periodic for all $q\in Q_L$. As usual,
    we denote by $\iota_{\mathbf{u}(q)}(p_X)$ (resp.
    $\pi_{\mathbf{u}(q)}(p_X)$) the preperiod (resp. the period) of
    $(\mathbf{u}_i(q)\mod p_X)_{i\ge 0}$.  We set
                $$I(p_X):=\max_{q\in Q_L} \iota_{\mathbf{u}(q)}(p_X)$$ 
                and
                $$P(p_X):=\lcm_{q\in Q_L} \pi_{\mathbf{u}(q)}(p_X).$$
                For $a_X$ large enough, we have $|\rep_S(a_X-1)|>d\,
                \# Q_L$.  By the pumping lemma applied to the product
                automaton\footnote{The automaton
                  $\mathcal{A}\times\mathcal{M}_L$ is defined as
                  follows. For any state $(q,q')$ in the set of states
                  $Q\times Q_L$, when reading $a\in\Sigma$, one
                  reaches in $\mathcal{A}\times\mathcal{M}_L$ the
                  state $(\delta(q,a),\delta_L(q',a))$.  The initial
                  state is $(q_0,q_{0,L})$ and the set of final states
                  is $F\times F_L$.  Roughly speaking, the product
                  automaton mimics the behavior of both automata
                  $\mathcal{A}$ and $\mathcal{M}_L$.}
                $\mathcal{A}\times\mathcal{M}_L$, there exist $x,y,z$
                with $y\neq\varepsilon$, $|xy|\le d\, \# Q_L$,
                $\delta(q_0,x)=\delta(q_0,xy)$,
                $\delta_L(q_{0,L},x)=\delta_L(q_{0,L},xy)$ and such
                that
                $$\rep_S(a_X-1)=xyz$$  and
                for all $n\ge 0$,
\begin{equation}
    \label{eq:arg}
xy^nz\in\rep_S(X).    
\end{equation}
Since $|xy|$ is bounded by a constant, we also have $|z|>I(p_X)$ if
$a_X$ is chosen large enough.

Since $|z|>I(p_X)$, using \eqref{eq:lelemme}, \eqref{eq:beta} and for
all $q\in Q_L$ the periodicity of the sequences $(\mathbf{u}_i(q)\mod
p_X)_{i\ge 0}$, we have for all $\ell\ge 0$ that
\begin{equation}
    \label{eq:arg2}
    \val_S(xy^{\ell p_X P(p_X)}yz)\equiv\val_S(xyz)\mod p_X.
\end{equation}
Let us give some extra details on how we derive 
identity \eqref{eq:arg2}. Assume $x=x_1\cdots x_r$, $y=y_1\cdots y_s$ and
$z=z_1\cdots z_t$. For all $n\ge 1$, using \eqref{eq:lelemme} for
$w=xy^nz$, we get $|w|=r+ns+t$ and 
\begin{eqnarray*}
\val_S(xy^nz)&=&
\sum_{q\in Q_L}\left(
\sum_{i=1}^{r}\beta_{q,i}(w) \mathbf{u}_{|w|-i}(q)\right. \\
&&
+\sum_{i=r+1}^{r+s} \beta_{q,i}(w) \mathbf{u}_{|w|-i}(q)
+\cdots + \sum_{i=r+(n-1)s+1}^{r+n s} \beta_{q,i}(w) \mathbf{u}_{|w|-i}(q)\\
&&\left.
+ \sum_{i=r+ns+1}^{r+n s+t} \beta_{q,i}(w) \mathbf{u}_{|w|-i}(q)
\right),
\end{eqnarray*}
where the first (resp. second, third) line corresponds, as explained
below, to the contribution of $x$ (resp. $y^n$, $z$).  By definition
\eqref{eq:beta} of the coefficients $\beta_{q,i}(w)$, we know that
$\beta_{q,1}(w)$ depends only on $x_1$ but $\beta_{q,2}(w)$ depends
only on $x_2$ and on $\delta_L(q_{0,L},x_1)$. Continuing this way,
$\beta_{q,r}(w)$ depends only on $x_r$ and on
$\delta_L(q_{0,L},x_1\cdots x_{r-1})$ and for $1\le j\le s$,
$\beta_{q,r+j}(w)$ depends on $y_j$ and on
$\delta_L(q_{0,L},xy_1\cdots y_{j-1})$. Now $\beta_{q,r+s+1}(w)$
depends only on $y_1$ and on $\delta_L(q_{0,L},xy_1\cdots
y_s)=\delta_L(q_{0,L},xy)=\delta_L(q_{0,L},x)$. This implies that
$\beta_{q,r+s+j}(w)=\beta_{q,r+j}(w)$ for all $q\in Q_L$ and all
$j\in\{1,\ldots,s\}$. This argument can be repeated with every copy of
$y$ appearing in $w$. Consequently, the previous expansion becomes
\begin{eqnarray*}
\val_S(xy^nz)&=&
\sum_{q\in Q_L}\biggr(
\sum_{i=1}^{r}\beta_{q,i}(w) \mathbf{u}_{|w|-i}(q)
+\sum_{i=r+1}^{r+s} \beta_{q,i}(w) \underbrace{\sum_{j=0}^{n-1}\mathbf{u}_{|w|-i-js}(q)}_{(*)}
 \\
&& + \sum_{i=r+ns+1}^{r+n s+t} \beta_{q,i}(w) \mathbf{u}_{|w|-i}(q)
\biggr).
\end{eqnarray*}
Assume now that $n=1+\ell p_X P(p_X)$, with $\ell\ge 0$. For $q\in Q_L$ and
$i=r+1,\ldots,r+s$, we have
$$(*)=\sum_{j=0}^{n-1}\mathbf{u}_{|w|-i-js}(q)=\mathbf{u}_{|w|-i}(q)+\sum_{j=1}^{\ell p_X P(p_X)}\mathbf{u}_{|w|-i-js}(q)$$
and the second term is congruent to $0$ modulo $p_X$ due to the
periodicity of the sequences $(\mathbf{u}_i(q)\mod p_X)_{i\ge 0}$
(recall that in the case we are considering, $|z|=t>I(p_X)$).
Consequently, for $n=1+\ell p_X P(p_X)$, we have
\begin{eqnarray*}
\val_S(xy^nz)&\equiv&
\sum_{q\in Q_L}\biggr(
\sum_{i=1}^{r}\beta_{q,i}(w) \mathbf{u}_{|w|-i}(q)
+\sum_{i=r+1}^{r+s} \beta_{q,i}(w) \mathbf{u}_{|w|-i}(q)
 \\
&& + \sum_{i=r+ns+1}^{r+n s+t} \beta_{q,i}(w) \mathbf{u}_{|w|-i}(q)
\biggr) \mod p_X.
\end{eqnarray*}
It is then easy to derive \eqref{eq:arg2}.
\medskip

We now use the minimality of $a_X$ to get a contradiction. Assume that
$a_X-1$ is in $X$ (the case not in $X$ is similar). Therefore for all
$n\ge 1$, $a_X+np_X-1$ is not in $X$. From \eqref{eq:arg}, for $\ell>0$
we get $xy^{\ell p_XP(p_X)}yz\in\rep_S(X)$, but from \eqref{eq:arg2} this
word represents a number of the kind $a_X+np_X-1$ with $n>0$ which
cannot belong to $X$.

Notice that $C$ can be effectively estimated as follows. One has to
choose a constant $C$ such that $a_X>C$ implies $|\rep_S(a_X-1)|-d\,
\#Q_L>I(p_X)$.  Since the abstract numeration system $S$, the period
$p_X$ and the number $d$ of states are given, $I(p_X)$ and $\rep_S(n)$
for all $n\ge 0$ can be effectively computed.
\end{proof}

\begin{theorem}\label{the:ans}
Let $S=(L,\Sigma,<)$ be an abstract numeration system such that for
    all states $q$ of the trim minimal automaton
    $\mathcal{M}_L=(Q_L,q_{0,L},\Sigma,\delta_L,F_L)$ of $L$
    $$\lim_{i\to\infty}\mathbf{u}_i(q)=+\infty$$
    and
    $\mathbf{u}_i(q_{0,L})>0$ for all $i\ge 0$. Assume moreover that
    $\mathbf{v}=(\mathbf{v}_i(q_{0,L}))_{i\ge 0}$ is such that
    $\lim_{m\to+\infty}N_{\mathbf{v}}(m)=+\infty$.  It is decidable
    whether or not an $S$-recognizable set is ultimately periodic.
\end{theorem}

\begin{proof}
    
    The proof is essentially the same as the one of
    Theorem~\ref{the:finfin}. First notice that the sequence
    $\mathbf{v}=(\mathbf{v}_i(q_{0,L}))_{i\ge 0}$ satisfies a linear
    recurrence relation of the kind \eqref{eq:linrec} having $a_k$ as
    last coefficient. Moreover it is increasing, since
    $\mathbf{u}_i(q_{0,L})>0$ for all $i\ge 0$.

    Let the prime decomposition of $|a_k|$ be $|a_k|=p_1^{u_1}\cdots
    p_r^{u_r}$ with $u_1,\ldots,u_r> 0$. Consider a DFA $\mathcal{A}$
    with $d$ states accepting an $S$-recognizable set
    $X\subseteq\mathbb{N}$. Assume that $X$ is periodic with period
    $$p_X=p_1^{v_1}\cdots p_r^{v_r}c$$
    where $\gcd(a_k,c)=1$ and
    $v_1,\ldots,v_r\ge 0$.
    
    By Proposition~\ref{pro:absnum}, we get $N_{\mathbf{v}}(p_X)\le d\#Q_L$.  
    Using Remark~\ref{rem:nupiu}, we obtain
                $$N_{\mathbf{v}}(c)\le\pi_{\mathbf{v}}(c)
                \le\pi_{\mathbf{v}}(p_X)\le(N_{\mathbf{v}}(p_X))^k\le(d\#Q_L)^k.$$
    Let $\alpha(m)$ be
    defined as the largest index $i$ such that
    $\mathbf{v}_i(q_{0,L})<m$.  Notice that since the
    sequence $\mathbf{v}$ is increasing, the map $m\mapsto\alpha(m)$
    is non-decreasing and $\lim_{m\to+\infty}\alpha(m)=+\infty$. Since
    $\gcd(a_k,c)=1$, the sequence $(\mathbf{v}_i(q_{0,L})\mod c)_{i\ge
      0}$ is purely periodic and $N_\mathbf{v}(c)\ge \alpha(c)$.
    Therefore, $\alpha(c)\le(d\# Q_L)^k$ and we can give effectively
    an upper bound on $c$.
    
    Now we can give upper bound on the $v_j$'s. The assumption
    $\lim_{m\to+\infty}N_{\mathbf{v}}(m)=+\infty$ implies that
    $\lim_{v\to+\infty}N_{\mathbf{v}}(p_j^v)=+\infty$ and we have
    exactly the same reasoning as in the proof of
    Theorem~\ref{the:finfin}.
    
    We have shown that if $X$ is ultimately periodic, then its period
    $p_X$ is bounded by a constant that can be effectively estimated.
    Using Proposition~\ref{pro:ans}, its preperiod is bounded by a
    constant which can also be computed effectively.

    Consequently, the sets of admissible periods and preperiods we
    have to check are finite.  Thanks to Proposition~\ref{pro:pa}, one
    has to build an automaton for each ultimately periodic set
    corresponding to a pair of admissible preperiods and periods and
    then compare the accepted language with $\rep_S(X)$.
\end{proof}

\begin{example}
    The abstract numeration system given in Example~\ref{ex:2fib}
    satisfies all the assumptions of the previous theorem. 
\end{example}

Theorem~\ref{the:ans} can be used to decide particular instances of
the HD0L periodicity problem. Let $\Delta$, $\Gamma$ be two finite
alphabets. Consider two morphisms $f:\Delta\to\Gamma^*$ and
$g:\Delta\to\Delta^*$ such that $g$ is prolongeable on a letter $a$.
The question is to decide whether or not the infinite word
$f(g^\omega(a))=w_0w_1w_2\cdots$ is ultimately periodic. Thanks to
\cite{RM}, one can canonically build an abstract numeration system
$S=(L,\Sigma,<)$ and a deterministic finite automaton with output
$\mathcal{M}=(Q,q_0,\Sigma,\delta,\Gamma,\tau)$ where
$\tau:Q\to\Gamma$ is the output function such that
$$\forall n\ge 0,\ w_n=\tau(\delta(q_0,\rep_S(n))).$$
Such a sequence
is said to be an {\em $S$-automatic sequence}. Notice that $(w_n)_{n\ge
  0}$ is ultimately periodic if and only if for all $b\in\Gamma$,
the $S$-recognizable set
$$X_b=\{n\mid w_n=b\}$$
is ultimately periodic. If $f$ and $g$ are
such that the associated numeration system $S$ satisfies the
assumptions of Theorem~\ref{the:ans}, then one can decide whether or
not $X_b$ (and therefore $(w_n)_{n\ge 0}$) is ultimately periodic.

\section*{Acknowledgments}
We warmly thank Jacques Sakarovitch for casting a new light on Juha
Honkala's paper. We also thank Christiane Frougny, J\'er\^ome Leroux
and Narad Rampersad for pointing out some useful references.

\end{document}